%% file: main.tex
\documentclass[11pt]{llncs}
\usepackage{packages}
\usepackage{macros}
\usepackage{parskip}

\usepackage{todonotes}

\title{MIRA: a Digital Signature Scheme based on the MinRank problem and the MPC-in-the-Head paradigm}
\author{Nicolas Aragon \inst{1} \and Loïc Bidoux\inst{2} \and Jesús-Javier Chi-Domínguez\inst{2} \and Thibauld Feneuil\inst{3,4} \and\\
Philippe Gaborit\inst{5} \and Romaric Neveu\inst{5} \and Matthieu Rivain\inst{4}}
\institute{
  Naquidis Center, Talence, France \and
  Technology Innovation Institute, UAE \and
  Sorbonne Université, CNRS, INRIA, Institut de Mathématiques \\ de Jussieu-Paris Rive Gauche, Ouragan, Paris, France \and
  CryptoExperts, Paris, France \and
  University of Limoges, France
}
\begin{document}
\let\oldaddcontentsline\addcontentsline
\def\addcontentsline#1#2#3{}
\maketitle
\def\addcontentsline#1#2#3{\oldaddcontentsline{#1}{#2}{#3}}
\input{0-abstract}
\tableofcontents
%--------------------------------------------------------------------%
\newpage
\section{Introduction}
\input{1-intro}
%--------------------------------------------------------------------%
\section{Preliminaries} \label{preliminaries}
%--------------------------------------------------------------------%
\subsection{Notations and Conventions}
\input{2-1-0-notations}
%--------------------------------------------------------------------%
\subsubsection{Pseudorandom Generators}
\input{2-1-1-PRG}
%--------------------------------------------------------------------%
\subsubsection{Collision-Resistant Hash Functions}
\hfill\\
\input{2-1-2-hashfct}
%--------------------------------------------------------------------%
\subsubsection{Commitments Schemes} 
\hfill\\
\input{2-1-3-cmtschemes}
%--------------------------------------------------------------------%
\subsubsection{Merkle Trees} 
\hfill\\
\input{2-1-4-MerkleTree}
%--------------------------------------------------------------------%
\subsubsection{Secret Sharing Schemes}
\hfill\\
\input{2-1-5-standardprimitives}
%--------------------------------------------------------------------%
\subsection{Proof of Knowledge and Digital Signature Schemes}
%--------------------------------------------------------------------%
\subsubsection{Zero-Knowledge Proof of Knowledge}
\hfill \\
\input{2-2-1-zk_pok}
%--------------------------------------------------------------------%
\input{2-2-1-dss}
%--------------------------------------------------------------------%
\subsubsection{Fiat-Shamir Transformation}
\hfill \\
\input{2-2-2-fst}

\subsubsection{Useful Lemmas}
\hfill \\
\input{2-2-3-lemmas}
%--------------------------------------------------------------------%
\subsubsection{q-Polynomials and Rank Metric definition}\hfill \\
\input{2-2-4-qpol}
%--------------------------------------------------------------------%
\subsection{MPC-in-the-Head and Proof of Knowledge} \label{mpc_pok}
\input{2-3-mpcith_pok}

%--------------------------------------------------------------------%
\subsection{The MinRank Problem}
\input{2-4-MinRank}

%--------------------------------------------------------------------%
\subsection{Rank Checking MPC protocol}
\input{2-5-rank_check}

%--------------------------------------------------------------------%
\section{Description of the Protocols} \label{specif}
\input{3-0-Proof_MinRank}

%--------------------------------------------------------------------%
\section{Parameter Sets} \label{parameters}
%--------------------------------------------------------------------%
\subsection{Parameters Choice}
\input{4-1-parameters}
%--------------------------------------------------------------------%
\subsection{Key and Signature Sizes}
\subsubsection{Additive Protocol Signature Size}\hfill \\

\input{4-2-hypercubesignsize}

\subsubsection{Threshold Protocol Signature Size}\hfill \\

\input{4-3-thresholdsignsize}
%--------------------------------------------------------------------%
\section{Security Analysis} \label{sec_proofs}
%--------------------------------------------------------------------%
\subsection{Security Proofs for the Proofs of Knowledge} \label{sec_proofs_pok}
\input{5-1-1-security_proof_hypercube}

\input{5-1-2-security_proof_th}
%--------------------------------------------------------------------%
\subsection{Security proofs for the Signature Schemes}\label{sec_proofs_sig}
The proofs follow in large parts the proofs in \cite{AGHHJY22}, \cite{FJR22} and \cite{FR22}.
\input{5-2-1-security_proof_sig_hypercube}
\input{5-2-2-security_proof_sig_th}
%--------------------------------------------------------------------%
\section{Known Attacks} \label{atk}
%--------------------------------------------------------------------%
\subsection{Attacks against Fiat-Shamir Signatures} \label{atk_fs}
\input{6-1-attacksFS}
%--------------------------------------------------------------------%
\subsection{Attacks against MinRank} \label{atk_mr}
\input{6-2-attacksMinRank}

%--------------------------------------------------------------------%
\newpage
\bibliographystyle{alpha}
\bibliography{ref}
%--------------------------------------------------------------------%
\end{document}

%% file: 0-abstract.tex
%!TEX root = ./main.tex

\begin{abstract}
    We exploit the idea of \cite{F22} which proposes to build an efficient signature scheme based on a zero-knowledge proof of knowledge of a solution of a MinRank instance. The scheme uses the MPCitH paradigm, which is an efficient way to build ZK proofs. We combine this idea with another idea, the hypercube technique introduced in \cite{AGHHJY22}, which leads to more efficient MPCitH-based scheme. This new approach is more efficient than classical MPCitH, as it allows to reduce the number of party computation. This gives us a first scheme called MIRA-Additive. We then present an other scheme, based on low-threshold secret sharings, called MIRA-Threshold, which is a faster scheme, at the price of larger signatures. The construction of MPCitH using threshold secret sharing is detailed in \cite{FR22}. These two constructions allows us to be faster than classical MPCitH, with a size of signature around $5.6$kB with MIRA-Additive, and $8.3$kB with MIRA-Threshold. We detail here the constructions and optimizations of the schemes, as well as their security proofs. 
\end{abstract}

%% file: 1-intro.tex
MIRA is a signature scheme designed to be secure against attacks from a quantum computer. The scheme is based on the MPC-in-the-Head paradigm and its security relies on the hardness to solve the MinRank problem. The underlying proof of knowledge uses symmetric functions as ingredients, such as hash functions and commitment schemes.

In section \ref{preliminaries}, we remind the reader some notations and definitions. We remind as well the MPC protocol we will use. In section \ref{specif}, we describe the two variants of the MPCitH scheme, one using additive secret sharing, the other using threshold secret sharing. Then, section \ref{parameters} explains the choice of parameters, and the obtained theoretical sizes. Section \ref{sec_proofs} deals with the security proofs of both schemes. Finally, section \ref{atk} is dedicated to the security of the scheme regarding to the Fiat-Shamir transform and the MinRank problem. 

In our schemes, the MPC protocol we use is the linearized-polynomial protocol on MinRank described in \cite{F22}. 
The hypercube MPCitH idea comes from \cite{AGHHJY22}, while the threshold MPCitH one from \cite{FR22}.

%% file: 2-1-0-notations.tex
Let $A$ a randomized algorithm. We write $y\leftarrow A(x)$ the output of the algorithm for the output $x$. If $S$ is a set, we write $x\sampler S$ the uniform sampling of a random element $x$ in $S$. We write $x\samples{s} S$ the pseudo-random sampling in $S$ with seed $s$.

We denote the set of integers between $1$ and $N$ by $\oneto{N}$, which can be shortened in $[N]$.

We denote by $\mathbb{F}_q$ the finite field of order $q$. We use bold letters to denote vectors or matrices (for example, $\bm{u}\in\mathbb{F}_q^n$ and $u\in\mathbb{F}_q$).

A function $\mu : \mathbb{N} \rightarrow \mathbb{R}$ is said negligible if, for every positive polynomial $p(\cdot)$, there exists an integer $N_p > 0$ such that for every $\lambda > N_p$, we have $|\mu(\lambda)| < 1/p(\lambda)$. When not made explicit, a negligible function in $\lambda$ is denoted $negl(\lambda)$ while a polynomial function in $\lambda$ is denoted $poly(\lambda)$. We further use the notation
poly($\lambda_1$, $\lambda_2$, ...) for a polynomial function in several variables.

Two distributions $\{D_\lambda\}_\lambda$ and $\{E_\lambda\}_\lambda$ indexed by a security parameter $\lambda$ are $(t, \epsilon)$-indistinguishable (where $t$ and $\epsilon$ are $\mathbb{N} \leftarrow \mathbb{R}$ functions) if, for any algorithm $\mathcal{A}$ running in time at most $t(\lambda)$ we have
$$\prb[\mathcal{A}^{D_\lambda} () = 1] - \prb[\mathcal{A}^{E_\lambda}() = 1] \le \epsilon(\lambda) $$
with $\mathcal{A}^{Dist}$ meaning that $\mathcal{A}$ has access to a sampling oracle of distribution $Dist$. 

The two distributions are said \begin{itemize}
\item computationally indistinguishable if $\epsilon \in \mathsf{negl}(\lambda)$ for every $t\in \mathsf{poly}(\lambda)$;
\item statistically indistinguishable if $\epsilon \in \mathsf{negl}(\lambda)$ for every unbounded $t$;
\item perfectly indistinguishable if $\epsilon=0$ for every unbounded $t$.
\end{itemize}

%% file: 2-1-1-PRG.tex
\begin{definition}[Pseudorandom Generators]
Let $G : \{0,1\}^* \rightarrow \{ 0,1\}^*$, $\ell$ a polynomial such that $G(s) \in \{ 0,1\}^{\ell(\lambda)}$. G is a $(t,\epsilon)$-secure pseudorandom generator if: \begin{itemize}
    \item $\ell(\lambda) > \lambda$,
    \item the distributions $\{ G(s), s\leftarrow \{ 0,1\}^\lambda\}$ and $\{r,r \sampler \{ 0,1\}^{\ell(\lambda)}\}$ are indistinguishable.
\end{itemize}
\end{definition}
In some protocols, we are going to use TreePRG, which is a pseudorandom generator, which uses a root seed to generate $N$ other seeds in a structured way. This can be illustrated quite easily with the following figure: 

\begin{figure}[H]
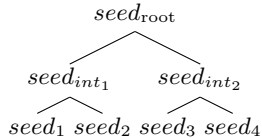

\Tree[.$seed_{\text{root}}$ [.$seed_{int_1}$ [.$seed_{1}$ ]
               [.$seed_{2}$ ]]
          [.$seed_{int_2}$ [.$seed_{3}$ ]
               [.$seed_{4}$ ]]]
               \captionof{figure}{Example of TreePRG}
               \label{TreePRG}
\end{figure}

Now, imagine one is looking to reveal $seed_{1}$, $seed_{3}$ , and $seed_{4}$, and hide $seed_{2}$. Then, all he has to do is reveal $seed_{int_2}$ and $seed_{1}$. It is impossible to retrieve $seed_{2}$, as we don't know the previous seed, but it is possible to retrieve the others. In this small example, we took $N=4$. This is especially interesting when, as in the additive-based MPCitH transformation we will see later, we reveal $N-1$ leaves. This means that we can do this operation by revealing only $\log_2(N)$ leaves instead of $N-1$. In general, for a TreePRG with $N$ final seeds and given a subset $I \subset \oneto{N}$, it is possible to reveal all the leaves but the ones in $I$ by revealing at most $|I|\log_2(\frac{N}{|I|})$ tree nodes.

We thus have three functions in order to deal with a TreePRG:
\begin{itemize}
    \item the root seed expansion, which generates $N$ seeds from a root seed;
    \item the sibling path derivation, which generates a sibling path from a leaf seed;
    \item the leaf seeds retrieval from a sibling path.
\end{itemize}

%% file: 2-1-2-hashfct.tex
We are going to use hash functions in the zero-knowledge protocols. We define below what is a collision resistant hash function.

\begin{definition}[Collision-Resistant Hash Function]
Let $h : \{0,1\}^* \rightarrow \{0,1\}^\lambda$. It is a collision-resistant hash function if $h$ can be computed in a polynomial time, and, for any polynomial algorithm $\mathcal{A}$, $$\prb[(x_1,x_2) \leftarrow \mathcal{A}(\lambda), x_1 \ne x_2, h(x_1) = h(x_2)] < \epsilon(\lambda)$$ where $\epsilon(\lambda)$ is negligeable.
\end{definition}

%% file: 2-1-3-cmtschemes.tex
The security of the proof of knowledge relies on commitments, in order to avoid an attacker to forge a valid transcript. The commitment scheme should satisfy two properties: the commitment should reveal no information about the data we committed (hiding property), and there should be only one way to open the commitment (binding property).

\begin{definition}[Commitment Scheme]
    A commitment scheme is defined by the function $\mathsf{Com}$, which takes as inputs ${m},\rho$ with $m$, and $\rho$, and outputs ${c}$, for some message space, randomness space, and commitment space.
        
\end{definition}

\begin{definition}[Hiding]
    A commitment scheme $Com$ is computationally (resp. statistically, resp. perfectly) hiding if, for every $m_0,m_1$, the distributions of $$\{ \mathsf{Com}(m_0,\rho), \rho \sample\}\text{ and }\{ \mathsf{Com}(m_1,\rho), \rho \sample\}$$ are computationally (resp. statistically, resp. perfectly) indistinguishable.
\end{definition}

\begin{definition}[Binding]
    A commitment scheme $Com$ is binding if, for every PPT algorithm $\mathcal{A}$, we have $$ \prb[\mathsf{Com}(m,\rho) = \mathsf{Com}(m',\rho') \cap m \ne m', (m,\rho,m',\rho') \leftarrow \mathcal{A}]<\mu(\lambda)$$ where $\mu$ is a negligible function. If we restrict $\mathcal{A}$ to a PPT, it is computationally binding. If the computation time is unbounded, it is statistically binding. 
\end{definition}

%% file: 2-1-4-MerkleTree.tex
A collision-resistant hash function (that we will note $\hash_M$) can be used to build a \textit{Merkle Tree}. Given inputs $v_1 \dots v_N$, we define $\mathsf{Merkle}(v_1 \dots v_N)$ as:  \begin{equation}
    \mathsf{Merkle}(v_1 \dots v_N) =
    \begin{cases}
    \hash_M\Big( {\mathsf{Merkle}(v_1 \dots v_{\frac{N}{2}}) } \| {\mathsf{Merkle}(v_{\frac{N}{2}+1} \dots v_N)}\Big) \text{ if } N>1 \\
    \hash_M(v_1) \text{ if } N=1 \\
    \end{cases}
\end{equation}

Thanks to Merkle Trees, similarly to the TreePRG described above, it is possible to verify that some given data is correct in an easy way. Given $I \subset \oneto{N}$, it is possible to verify that the $(v_i)_{i\in I}$ were indeed used to build the Merkle Tree only by revealing at most $|I|\log_2\Big(\frac{N}{|I|}\Big)$ hash values, by revealing the sibling paths of the $(v_i)_{i\in I}$. This path is called the authentification path and is denoted $\mathsf{auth}((v_1 \dots v_N),I)$.

%% file: 2-1-5-standardprimitives.tex
%!TEX root = ./main.tex

The interactive proof relies on a prover simulating a MultiParty Computation (MPC) protocol, where each party has a share of the witness $\bm{x}$. We detail here the formalism from \cite{FR22}. The sharing of a secret $s$ among $N$ parties is denoted $(\share{s}_1, \dots, \share{s}_N)$ where $\share{s}_i$ is the share of the $i$\textsuperscript{th} party. Given $J \subset \oneto{N}$, $\share{s}_J$ is the subset of shares $\{\share{s}_j\}_{j\in J}$.

\begin{definition}[Threshold Linear Secret Sharing]
Let $\mathbb{F}$ a finite field. Let $1<t\le N$. A $(t,N)$-threshold linear secret sharing scheme (TLSSS) is a scheme to share a secret, $s\in\mathbb{F}$, in a sharing $\share{s} = (\share{s}_1 \dots \share{s}_N) \in \mathbb{F}^N$, where only $t$ coordinates of $\share{s}$ need to be known in order to reconstruct the secret $s$, while the knowledge of $t-1$ coordinates of $\share{s}$ reveals no information.

We can write this the following way:
\begin{equation}
    \begin{cases}
    \mathsf{Share} : \mathbb{F} \times R \rightarrow \mathbb{F}^N \\
    \mathsf{Reconstruct}_J : \mathbb{F}^t \rightarrow \mathbb{F} \\
    \end{cases}
\end{equation}
$R$ corresponds to the randomness space used to build the shares. $J$ is a subset of $\{1 \dots N\}$, with $|J| = t$.
The two algorithms of a TLSSS must satisfy some properties: \begin{itemize}
\item \textbf{Correctness.} For every $s\in \mathbb{F},r\in R, J \subset \oneto{N}$ such that $|J|=t$ and for $\share{s} \leftarrow \mathsf{Share}(s;r) $, we have: 
$$\mathsf{Reconstruct}_J(\share{s}_J)=s$$
\item \textbf{Perfect ($t-1$)-privacy.} For every $s_0,s_1 \in \mathbb{F}, I\subset \oneto{N}$ with $|I|=t-1$, the two distributions: \begin{align*}\left\{ \share{s_0}_I \text{ }\Big\vert\text{ } \genfrac{}{}{0pt}{0}{r \sampler R}{\text{ } \share{s_0}_{\oneto{N}} \longleftarrow \mathsf{Share}(s_0,r)}\right\} \text{ and } \left\{ \share{s_1}_I \text{ }\Big\vert \text{ }\genfrac{}{}{0pt}{0}{r \sampler R}{\text{ } \share{s_1}_{\oneto{N}} \longleftarrow \mathsf{Share}(s_1,r)}\right\}\end{align*} are perfectly indistinguishable.
\item \textbf{Linearity.} For every $v_0,v_1 \in \mathbb{F}^t,\alpha \in \mathbb{F}, J \subset \oneto{N}$ with $|J| = t$, $$\reconstruct_J(\alpha\cdot v_0 + v_1) = \alpha \cdot \reconstruct_J(v_0)+\reconstruct(v_1).$$
\end{itemize}
\end{definition}

We recall below two among the most used secret sharing schemes. \\

\begin{definition}[Additive Secret Sharing]
Let $\mathbb{F}$ a field and $s \in \mathbb{F}$ a secret.
An additive secret sharing with $N$ parties is a $(N,N)$-threshold sharing scheme such that:
\begin{itemize}
\item $\share{s}_i = r_i$ for $i \in \oneto{N-1}$, where $r_i \sampler \mathbb{F}$; 
\item $\share{s}_N = s-\sum_{i=1}^{N-1}\share{s}_i$.
\end{itemize}
The $\mathsf{Reconstruct}_{\oneto{N}}$ algorithm takes as inputs all the shares, and outputs the sum of all the shares.
\end{definition}

\begin{definition}[Shamir's Secret Sharing]\label{def:shamir-sharing}
Let $\mathbb{F}$ a field and $s \in \mathbb{F}$ a secret.
A Shamir's secret sharing is the following $(\ell+1,N)$-threshold sharing scheme: \begin{itemize}
\item Sample $(r_1, \ldots, r_\ell) \sampler \mathbb{F}^\ell$;
\item Compute $P(X) = s + \sum_{i=1}^\ell r_iX^i$;
\item Compute $\share{s}_i = P(e_i)$ where the $(e_i)_{i\in \{1, \dots, N\}}$ are distinct and non-zero public values.
\end{itemize}
For $J$ a subset of $\oneto{N}$ with $|J| = \ell+1$, the $\mathsf{Reconstruct}_J$ algorithm corresponds to the interpolation of the polynomial P, when taking in inputs $\share{s}_i$ for $i\in J$, and outputs the constant term, $s$. \\
\end{definition}

\begin{proposition}
    Let an $(\ell+1,N)$-threshold LSSS. For each $v\in\mathbb{F}^{\ell+1}$ and each subset $J\subset\oneto{N}$ of $\ell+1$ elements, there exists a unique sharing $\share{x}_{\oneto{N}}\in\mathbb{F}^{N}$ such that $\share{x}_J=v$, and such that for all $\mathcal{J}\subset\oneto{N}$ of $\ell+1$ elements: $$\reconstruct_\mathcal{J}(\share{x}_\mathcal{J})=\reconstruct_J(v)$$
\end{proposition}

\begin{proof}
    See \cite{FR22}.
\end{proof}

One deduces there exists an algorithm $\expand_J$ which returns the unique sharing from a subset $J$ of the shares. For example, in the Shamir's secret sharing, $\expand_J$ builds the Lagrange polynomial from the known evaluations and outputs the image of each party's point.

%% file: 2-2-1-zk_pok.tex
We define here the concept of proof of knowledge. Let $R\subset \{ 0,1\}^*\times  \{ 0,1\}^*$ an NP-relation. $(x,\omega)\in R$ is a statement-witness pair where $x$ is the statement and $\omega$ an associated witness. The set of valid witnesses for a statement $x$ is $R(x)=\{\omega: (x,\omega)\in R\}$. A prover $\mathcal{P}$ wants to use a proof of knowledge to convince a verifier $\mathcal{V}$ that he knows a witness $\omega$ for a statement $x$.

\begin{definition}[Proof of knowledge]
    A proof of knowledge for a relation $R$ with soundness $\epsilon$ is a two-party protocol between a prover $\mathcal{P}$ and a verifier $\mathcal{V}$ with a public statement $x$, where $\mathcal{P}$ want to convince $\mathcal{V}$ that he knows $\omega$ such that $(x,\omega)\in R$. We denote $\lbrace \mathcal{P}(x,\omega),\mathcal{V}(x)\rbrace$ the transcript between $\mathcal{P}$ and $\mathcal{V}$. A proof of knowledge has the following properties: 
    \begin{itemize}
        \item \textbf{Perfect completeness: } If $(x,\omega) \in \mathcal{R}$, then a prover $\mathcal{P}$ who knows a witness $\omega$ for $x$ succeeds in convincing the verifier $\mathcal{V}$ of his knowledge. This means that the prover convinces the verifier with probability 1, i.e, $$\prb[\langle \mathcal{P}(x,\omega),\mathcal{V}(x)\rangle \rightarrow \text{ACCEPT}] = 1.$$
        \item \textbf{Soundness: } If there exists a PPT prover $\Tilde{\mathcal{P}}$ such that $$\Tilde{\epsilon} = \prb[\langle\Tilde{\mathcal{P}}(x),\mathcal{V}(x)\rangle \rightarrow \text{ACCEPT}] > \epsilon,$$ then there exists an algorithm which, given rewindable black-box access to $\Tilde{\mathcal{P}}$, outputs a witness $\omega'$ for x in time $\mathsf{poly}(\lambda,(\Tilde{\epsilon}-\epsilon)^{-1})$ with probability at least $\frac{1}{2}$.
    \end{itemize}
\end{definition}

To be zero-knowledge while the prover interacts with a honest verifier (i.e., a verifier sending his messages according to the definition of the protocol), the PoK must verify the following property :
\begin{definition}[Honest-Verifier Zero-Knowledge]
    A PoK satisfies the Honest-Verifier Zero-Knowledge (HZVK) property if there exists a polynomial-time simulator $\mathsf{Sim}$ that given as input a statement $x$ and random challenges $(\ch_1,...,\ch_n)$, outputs a transcript $\lbrace \mathsf{Sim}(x,\ch_1,...,\ch_n),\mathcal{V}(x)\rbrace$ which is computationally indistinguishable from the probability distribution of transcripts of honest executions between a prover $\mathcal{P}(x,w)$ and a verifier $\mathcal{V}(x)$.
\end{definition}

%% file: 2-2-1-dss.tex
%!TEX root = ./main.tex

\begin{definition}[Digital Signature Scheme]
    A digital signature scheme $\mathsf{DSS}$ with security parameter $\lambda$ is a triplet of polynomial time algorithms $(\mathsf{KeyGen}, \mathsf{Sign}, \mathsf{Verif})$ such that:\begin{itemize}
        \item The key-generation algorithm $\mathsf{KeyGen}$ is a probabilistic algorithm which outputs a pair of keys $(\mathsf{pk},\mathsf{sk})$.
        \item The signing algorithm $\mathsf{Sign}$, eventually probabilistic, which takes as inputs a message $m$ to sign and the secret key $\mathsf{sk}$, and outputs a signature $\sigma$.
        \item The verification algorithm $\mathsf{Verif}$ which takes as inputs the public key $\mathsf{pk}$, the message $m$ and its signature $\sigma$, and outputs a bit $b$. The output $1$ means that the signature is considered as valid.
    \end{itemize}
\end{definition}

A correct signature scheme satisfies the following property: if $(\pk,\sk)\leftarrow\mathsf{KeyGen}$, for all messages $m$ signed by $\sigma\leftarrow\mathsf{Sign}(\pk,m)$, we have $1\leftarrow\mathsf{Verif}(\sk,m,\sigma)$. This means, if a signature is correctly generated, then it is always accepted.

The standard security notion for digital signature schemes is existential unforgeability under adaptive chosen message attacks (EUF-CMA) is defined as follows: 
\begin{definition}[EUF-CMA]
    We can define the following game $G_{\textsc{euf-cma}}$ where ${\mathsf{Sign}(sk,\cdot)}$ is an oracle that sign messages:
    \begin{align*}
        &\textsl{(\pk,\sk)} \leftarrow \mathsf{KeyGen}() \\
        &(m^*,\sigma^*) \leftarrow \mathcal{A}^{\mathsf{Sign}(\sk,\cdot)}(\textsf{pk}) \\
        &(m,\sigma) \leftarrow \mathcal{A}
    \end{align*}
    The game returns $1$ if $\mathsf{Verif}(m,\sigma,\textsf{pk}) = 1$ and $m$ was not queried to ${\mathsf{Sign}(\sk,\cdot)}$.
    The signature scheme is EUF-CMA secure if, for every polynomial adversary $\mathcal{A}$, $\prb[G_{\textsc{euf-cma}}(\mathcal{A}) =1]$ is negligible.
\end{definition}

%% file: 2-2-2-fst.tex
%!TEX root = ./main.tex

The Fiat-Shamir (FS) transformation is a generic process allowing to convert an interactive identification scheme into a signature. The main adaptation lies in the removal of the interactions in the protocol: one needs to pull the challenges in a deterministic way to sign a message without the assistance of a verifier. Note that the protocol must be repeated several times to achieve the desired level of security. We note $\tau$ the number of repetitions. One will see in section~\ref{atk_fs} that there is an effective attack if $\tau$ is too small.

Let us describe the FS tranformation on a 5-round zero-knowledge proof of knowledge.
Concretely, the prover begins as in the zero-knowledge proof by committing all the auxiliary information. At the end of the first step, the prover computes: $$h_1=\hash_1(\salt, m, (h^{(e)}_0)_{e \in \oneto{\tau}})$$
where $\hash_1$ is an hash function, $h^{(e)}_0$ is the first step commitment of the $e^{\text{th}}$ execution of the protocol, and $\salt$ a random value in $\{0, 1\}^{2\lambda}$. The prover obtains the first challenge (one challenge per execution $e \in [1, \tau]$) from $h_1$ by using a XOF (Extendable Output Function).

The second challenge is generated in a similar way: the prover computes an element $h_2$ thanks to an other hash function $\hash_2$ and the other information computed during the step 3. The prover obtains the second challenge (one challenge per execution $e \in [1, \tau]$) from $h_2$ by using a XOF as well.

The signature $\sigma$ therefore consists of sending:
\begin{itemize}
    \item the salt $\salt$;
    \item $h_1$ as commitment of the initial values;
    \item $h_2$ as hash of all responses of the first challenge;
    \item each response $\rsp$ of the second challenge.
\end{itemize}
The signature is:
$$\sigma = (\salt,h_1, h_2, (\rsp^{(e)})_{e \in \oneto{\tau}}).$$

%% file: 2-2-3-lemmas.tex
\begin{lemma}[Splitting Lemma] \label{splitting_lemma}
Let $A \subset X\times Y$ such that $\prb[(x,y)\in A] \ge \epsilon$.
For any $\alpha<\epsilon$, let us define $$B = \Big\{  (x,y)\in X\times Y \mid \prb_{y'\in Y}[(x,y')\in A] \ge \epsilon-\alpha  \Big\} \text{ and } \Bar{B} = (X\times Y)\setminus B.$$ Then we have 
\begin{itemize}
    \item $\prb[B]\ge \alpha$,
    \item $\forall (x,y) \in B, \prb_{y' \in Y}[(x,y') \in A] \ge \epsilon - \alpha$,
    \item $\prb[B | A] \ge \frac{\alpha}{\epsilon}$.
\end{itemize}
\end{lemma}
\begin{proof}
    An interested reader can refer to \cite{PS00}.
\end{proof}

%% file: 2-2-4-qpol.tex
%!TEX root = ./main.tex

One of the main tools we are going to use are q-polynomials, as they allow us to characterize linear subspaces.
\begin{definition}[q-polynomial]
A q-polynomial of q-degree $r$ is a polynomial in $\Fqm[X]$ of the form:
$$P(X) = X^{q^r}+\sum_{i=0}^{r-1}p_i \cdot X^{q^i} \qquad \text{with } p_i \in \Fqm.$$
\end{definition}

\begin{proposition}
Let P a q-polynomial, $\alpha,\beta \in \Fq$, $x,y\in \Fqm$. We then have:
$$P(\alpha x + \beta y )=\alpha P(x) + \beta P(y)$$
\end{proposition}
\begin{proof}
This comes directly from the fact that the Frobenius endomorphism: $x \xmapsto[]{} x^q$ is linear over $\Fq$
\end{proof}
We can see q-polynomials as $\mathbb{F}_q$-linear applications from $\mathbb{F}_{q}^m$ to $\Fqm$. It is then possible to define a linear subspace of $\Fqm$ from a q-polynomial.
\begin{proposition}
    The set of roots of a non-zero q-polynomial of q-degree r forms a linear subspace of dimension lower than or equal to $r$.
\end{proposition}
\begin{proof}
    Let $P$ a q-polynomial of degree $r$. One can see $P$ as a linear application from $\Fq^m$ to $\Fq^m$. As an endomorphism kernel, the set of zeros forms a linear space. Since it is a polynomial of degree $q^r$, $P$ has at most $q^r$ roots, giving the upper bound on the subspace dimension.
\end{proof}

\begin{proposition}[\cite{ore}]
Let $E$ a linear subspace of $\Fqm$ of dimension $r\leq m$. Then there exists a unique monic q-polynomial of q-degree $r$ such that every element in $E$ is a root of $P$.\\
$P$ is called the annihilator polynomial of $E$.
\end{proposition}

%The proof is not trivial. The interested reader can read \cite{ore}.

Let us define the usual notions of the rank metric:
\begin{definition}
    Let $\bm{E} \in \mathbb{F}_{q}^{m \times n} = \Big(e_{i,j}\Big)$ with $e_{i,j} \in \Fq$ for $(i,j)\in \oneto{m} \times \oneto{n}$, and let $\mathcal{B} = ( b_1 \dots b_m )$ an $\Fq$-basis of $\Fqm$. It is then possible to associate each column of $\bm{E}$ to an element of $\Fqm$ using $$e_j = \sum_{i=1}^{m}b_ie_{i,j}$$ for each $j \in \oneto{n}$.
    By defining $\bm{e} = (e_1 \dots e_n)$, we can say that $\bm{e}$ is the vector associated to the matrix $\bm{E}$.

    The \emph{rank weight} is defined as $\operatorname{W}_R(\bm{e}) = \rank(\bm{E})$.
    The \emph{distance} between two vectors $\bm{x}$ and $\bm{y}$ is then $d(\bm{x},\bm{y}) = \operatorname{W}_R(\bm{x}-\bm{y})$.
    The \emph{support} of $\bm{e} = (e_1 \dots e_n)$ is the linear subspace of $\Fqm$ generated by its coordinates: $\supp(\bm{e}) = \langle e_1 \dots e_n \rangle$.
\end{definition}

\begin{remark}
    The choice of the basis $\mathcal{B}$ does not change anything to the weight of $\bm{e}$ or the rank of $\bm{E}$. The rank of $\bm{E}$ is obviously equal to the dimension of $\supp(\bm{e})$. Moreover, in the above definition, we are working on the columns of $\bm{E}$. We stress that it is possible to work with the rows instead, as it may be more efficient, depending on the values of $m$ and $n$. We will also abuse notations, and sometimes note $\supp(\bm{e})$ as $\supp(\bm{E})$.
\end{remark}

The number of supports possible of dimension $r$ when working in $\Fqm$ is $$\CG{m}{r} = \prod_{i=0}^{r-1}\frac{q^m-q^i}{q^r-q^i} \approx q^{r(m-r)}.$$

%% file: 2-3-mpcith_pok.tex
%!TEX root = ./main.tex

The following explanation of the MPCitH paradigm comes from \cite[Section 2.1]{F22}.

The MPC-in-the-Head (MPCitH) paradigm introduced in \cite{IKO} offers a way to build zero-knowledge proofs from secure multi-party computation (MPC) protocols. Let us assume we have an MPC protocol in which $N$ parties $\party_1, \ldots, \party_N$ securely and correctly evaluate a function $f$ on a secret input $w$ with the following properties:
\begin{itemize}
    \item the secret witness $w$ is encoded as a sharing $\share{w}$ and each $\party_i$ takes a share $\share{w}_i$ as input;
    \item the function $f$ outputs $\textsc{Accept}$ or $\textsc{Reject}$;
    \item the views of $t$ parties leak no information about the secret $w$, where $t+1$ is the threshold of the secret sharing.
\end{itemize}
We can use this MPC protocol to build a zero-knowledge proof of knowledge of a witness $w$ for which $f(w)$ evaluates to \textsc{Accept}. The prover proceeds as follows:
\begin{itemize}
    \item she builds a random sharing $\share{w}$ of $w$;
    \item she simulates locally (``in her head'') all the parties of the MPC protocol;
    \item she sends commitments to each party's view, i.e, party's input share, secret random tape and sent and received messages, to the verifier;
    \item she sends the output shares $\share{f(w)}$ of the parties, which should correspond to \textsc{Accept}.
\end{itemize}
Then the verifier randomly chooses $t$ parties and asks the prover to reveal their views. After receiving them, the verifier checks that they are consistent with an honest execution of the MPC protocol and with the commitments. Since only $t$ parties are opened, revealed views leak no information about the secret $w$, while the random choice of the opened parties makes the cheating probability upper bounded by $(N-t)/N$, thus ensuring the soundness of the zero-knowledge proof.\footnote{We implicitly assume here that the communication between parties is broadcast.}%
\bigskip

In our case, the parties take as input a linear sharing $\share{w}$ of the secret $w$ (one share per party) and they compute one or several rounds in which they perform three types of actions:
\begin{description}
    \item[Receiving randomness:] the parties receive a random value $\epsilon$ from a randomness oracle $\oracle_R$. When calling this oracle, all the parties get the same random value $\epsilon$.
    \item[Receiving hint:] the parties can receive a sharing $\share{\beta}$ (one share per party) from a hint oracle $\oracle_H$. The hint $\beta$ can depend on the witness $w$ and the previous random values sampled from $\oracle_R$.
    \item[Computing \& broadcasting:] the parties can locally compute $\share{\alpha} := \share{\varphi(v)}$ from a sharing $\share{v}$ where $\varphi$ is an $\mathbb{F}$-linear function, then broadcast all the shares $\share{\alpha}_1$, \ldots, $\share{\alpha}_N$ to publicly reconstruct $\alpha := \varphi(v)$. The function $\varphi$ can depend on the previous random values $\{\epsilon^i\}_i$ from $\oracle_R$ and on the previous broadcasted values. One should note that in the case of additive sharing, only one party needs to compute the addition by a constant.
\end{description}

We restrain here to threshold linear sharings, i.e, additive secret sharing and low-threshold linear secret sharing, which is the framework in which \cite{FR22} and \cite{F22} are set.  
\bigskip

In \cite{AGHHJY22}, another way to verify the parties' computation (in the case of additive sharing) has been introduced which is more efficient. This is the hypercube technique, which we will detail here.

The idea of the hypercube construction is the following:
\begin{itemize}
    \item Generate $N = 2^D$ parties with an additive sharing. Those parties can be indexed either by an integer in $\oneto{N}$ or by a vector in $\oneto{2}^D$. For example, the $i$-th leaf can be written as the leaf $i$, or as the leaf $(i_1,\dots,i_D)$;
    \item For each dimension of the hypercube, compute the main shares by summing up the $2^{D-1}$ shares of a slice of the hypercube. We have $2$ main parties per dimension, i.e, $2\cdot D$ main parties in total. We can also establish a mapping between the main parties and the leaves of the hypercube, by writing a main party $p = (k,j)$ $\in \oneto{D}\times \oneto{2}$. Concretely, for a witness $w$, the share of the main party $p = (k,j)$ can be built (and written) $\share{w}_{(k,j)}=\sum_{i: i_{k}=j}\share{w}_i$. These main shares will be written either as  $\share{w}_{(k,j)}$ or as $\share{w}_{(p_1,p_2)}$, depending on the indices we use.\\
    \item Execute the MPC protocol for all the dimension, i.e, execute $D$ times the MPC protocol, each time with the $2$ main parties of the dimension.
\end{itemize}
Since we built the hypercube with an additive sharing, we know that the sum of all the leaves gives us the secret. This means that we can sum them up in every way we want (the way in which we sum them depends on the dimension). This is an improvement compared to the standard MPC in the head, as instead of simulating one protocol with $256$ parties, we can simulate $8$ MPC protocols, with only 2 parties each. Furthermore, since the used secret is the same for every dimension and since the broadcasted plaintext values are the same for every protocol, it is possible to:
\begin{itemize}
    \item Execute the MPC protocol for the two parties for one dimension;
    \item For $D-1$ dimensions: 
    \begin{itemize}
        \item Execute the MPC protocol for one party $\mathcal{P}_1$;
        \item Subtract the share broadcasted by $\mathcal{P}_1$ to the broadcast plaintext value;
        \item Set the resulting share as the broadcasted share of the second party $\mathcal{P}_2$.
    \end{itemize}
\end{itemize}
Even though this doesn't reduce the size of the communication cost, this is an improvement as the computations are faster (we avoid the computation for $D-1$ parties). Moreover, we can observe that taking a hypercube with edge size larger than $2$ brings no advantage over a power of $2$ (see \cite{AGHHJY22}). This is why $N$ is always a power of $2$ in the case of the hypercube scheme.

%% file: 2-4-MinRank.tex
%!TEX root = ./main.tex

\begin{definition}[MinRank]
Let $\mathbb{F}_q$ be the finite field of size $q$, and $m,n,k,r \in \mathbb{N^*} $. The computational MinRank Problem with parameters $(q,m,n,k,r)$ is the following problem:\\
Let $\bm{M}_1, \dots, \bm{M}_k, \bm{E}   \in \mathbb{F}_{q}^{m\times n}$ and $\bm{x} \in \mathbb{F}_q^k$ be uniformly sampled such that
$$\operatorname{W}_R(\bm{E}) \leq r ~~~\text{with}~~~  \bm{M}_0 := \bm{E} - \sum_{i=1}^{k} x_i \bm{M}_i.$$
Given $\bm{M}_0 ,\dots, \bm{M}_k$, retrieve the vector $\bm{x}$.
\end{definition}

The MinRank problem was proven to be NP-Complete by \cite{BFS99}, and plays a central role in cryptography. It is used in the attacks on HFE, and appears in syndrome decoding in rank metric. It was also used in some signatures, such as Courtois' signature \cite{courtois}, MR-DSS \cite{BESV22}, or MinRank in the Head \cite{ARV22} for instance. 

\subsubsection{Key Generation} \hfill \\

In practice, to generate a MinRank instance, one has to: 
\begin{itemize}
    \item Sample $\bm{x}$ uniformly in $\Fqk$;
    \item for all $i \in \oneto{k}$, sample $\bm{M}_i \in \Fq^{m\times n}$;
    \item Sample $\bm{E} \in \Fq^{m\times n}$ such that $\operatorname{W}_R(\bm{E})\le r$;
    \item Define $\bm{M}_0$ as $\bm{E}-\sum_{i=1}^{k}\bm{M}_ix_i$.
\end{itemize}

From this, we can quickly explicit a $\mathsf{KeyGen}$ algorithm:
\begin{figure}[H]
\pcb[codesize=\scriptsize, minlineheight=0.75\baselineskip, mode=text, width=0.90\textwidth] {
    \begin{itemize}
        \item $\bm{x} \sampler \Fqk$
        \item For $i \in \oneto{k}$, $\bm{M}_i \sampler \Fq^{m\times n}$
        \item $\bm{E} \sampler  \Fq^{m\times n}$ such that $\operatorname{W}_R({\bm{E}}) \le r$
        \item $\bm{M}_0 = \bm{E}-\sum_{i=1}^{k}\bm{M}_ix_i$
        \item Set $\mathsf{pk} = (\bm{M}_0,\dots,\bm{M}_k)$ and $\mathsf{sk} = \bm{x}$
        \item Return $(\mathsf{pk},\mathsf{sk})$.
    \end{itemize}
}
\vspace{-\baselineskip}
\captionof{figure}{Key generation -- Simple version}
\end{figure}
The size of the public key will be $m\cdot n \cdot \log_2 q+\lambda$ bits ($\bm{M}_0$ and the seed used for the random matrices), and the size of the secret key will be $\lambda$ bits (if we sample $\bm{x}$ from a seed).

However, it is possible to optimize this key generation. First, we need to define the systematic form of a MinRank problem. Let $\bm{L}_1$ be the $(k+1) \times mn$ matrix, composed of $\bm{M}_0, \dots, \bm{M}_k$ written in lines, following the row order. Then, if the first $k$ columns and $k$ rows of $\bm{L}_1$ form a full rank matrix, we can obtain a matrix $\bm{L} = \begin{bmatrix}
    \begin{matrix}
    \bm{I}_k \\
    0 \dots 0
    \end{matrix} & \bm{L'}_1
\end{bmatrix}$ where $\bm{I}_k$ is the $k\times k$ identity matrix, using row operations. Row operations on the matrix $\bm{L}_1$ correspond to a linear combination of $\bm{M}_0, \dots, \bm{M}_k$, meaning they don't change the number of solutions of the instance.

For a generic MinRank instance (i.e, with high probability, as noted in \cite{BESV22} and used in \cite{BBBGT22}), this transformation will be possible. This means we can generate the instance directly in this systematic form, without changing the difficulty of the problem. To proceed, we will use the procedure described in~\cite{BESV22}:
\begin{itemize}
    \item Sample $\bm{L} \in \Fq^{k\times mn}$ uniformly such that $\bm{L} = \begin{bmatrix}
        \bm{I}_k & \bigg| L'
    \end{bmatrix}$
    where $\bm{I}_k$ is the $k\times k$ identity matrix and $\bm{L'}$ is an $k \times (mn-k)$ matrix;
    \item Set $\bm{M}_i$ as the i-th line of $\bm{L}$, written in a matrix form (in row order, i.e, the first row of the matrix is the first $n$ entries of the lines, the second row is the next $n$ entries and so on);
    \item Sample $\bm{E} \sampler \Fq^{m\times n}$ uniformly such that $\operatorname{W}_R(\bm{E}) \le r $ and $\bm{\beta} \sampler \Fqk$;
    \item Compute $\bm{F} = \bm{E}-\sum_{i=1}^{k}\beta_i\bm{M}_i$;
    \item Compute $\bm{M}_0 = \bm{F} - \sum_{i=1}^{k}f_i\bm{M}_i$ where $\bm{f} = (f_1 \dots f_k)$ are the first $k$ entries, in row order, of $\bm{F}$;
    \item Set $\bm{x} = \bm{\beta} + \bm{f}$.
\end{itemize}
It is clear that we have $\bm{M}_0 = \bm{E} - \sum_{i=1}^{k} x_i\bm{M}_i$ and that the first $k$ entries of $\bm{M}_0$ are zeros. Since $\bm{L}$ comes from a random seed, our secret key $\bm{x}$ can be retrieved from the used seed, and is thus of size $\lambda$ bits. For our public key, we have $\bm{M}_1, \dots, \bm{M}_k$ of size $\lambda$ bits, plus $\bm{M}_0$, of size $(mn-k)\cdot \log_2 q$ bits.
\begin{figure}[H]
\pcb[codesize=\scriptsize, minlineheight=0.75\baselineskip, mode=text, width=0.90\textwidth] { 
    \begin{itemize}
        \item Sample $\bm{L} \in \Fq^{k\times mn}$ uniformly such that $\bm{L}$ is of the form described above
        \item Set $\bm{M}_i$ as the i-th line of $\bm{L}$, written in a matrix form (in row order)
        \item Sample $\bm{E} \sampler \Fq^{m\times n}$ uniformly such that $\operatorname{W}_R(\bm{E}) \le r $, and $\bm{\beta} \sampler \Fqk$
        \item Compute $\bm{F} = \bm{E}-\sum_{i=1}^{k}\beta_i\bm{M}_i$
        \item Compute $\bm{M}_0 = \bm{F} - \sum_{i=1}^{k}f_i\bm{M}_i$ where $\bm{f} = (f_1 \dots f_k)$ are the first $k$ entries, in row order, of $\bm{F}$
        \item Set $\bm{x} = \bm{\beta} + \bm{f}$
        \item Set $\mathsf{pk} = (\bm{M}_0,\dots,\bm{M}_k)$ and $\mathsf{sk} = \bm{x}$
        \item Return $(\mathsf{pk},\mathsf{sk})$.
    \end{itemize}
}
\vspace{-\baselineskip}
\captionof{figure}{Key generation -- Optimized version}
\end{figure}

%% file: 2-5-rank_check.tex
%!TEX root = ./main.tex

It is necessary to build a protocol allowing us to verify that the space generated by a list of $n \geq r$ elements is at most $r$.
We will build here a MPCitH ZK-proof based on Feneuil's  protocol using q-polynomials \cite{F22}, that we will remind here.
We want to check that a list of elements $(e_j)_{j\in \oneto{n}}$, where $e_j\in\Fqm$, leads to a $\mathbb{F}_q$-linear subspace $U$ of dimension at most $r$.

Given $U$, let $L_U$ be the polynomial: $$L_U=\prod_{u\in U} (X-u)\in\Fqm[X].$$
It is possible to show that all the $e_j$ are roots of $L_U$, which is a q-polynomial. It means that $L_U$ is of the form: $$L_U(X)=\sum_{i=0}^{r-1}\beta_iX^{q^i}+X^{q^r}$$
Rather than checking separately that each $e_j$ is a root of $L_U$, we will batch all these verifications by uniformly sampling $\gamma_1,...,\gamma_n$ in an extension $\mathbb{F}_{q^{m\eta}}$ of $\Fqm$ and check that: $$\sum_{j=1}^n\gamma_jL_U(e_j)=0$$
If one or more of the $e_j$ is not a root of $L_U$, the above equation is satisfied with probability at most $\frac{1}{q^{m\eta}}$. Then, \begin{align*}
\sum_{j=1}^n\gamma_jL_U(e_j) &= \sum_{j=1}^n\gamma_j \left(\sum_{i=0}^{r-1}\beta_ie_j^{q^i}+e_j^{q^r}\right)\\
&= \sum_{j=1}^n\gamma_je_j^{q^r} +\sum_{i=0}^{r-1}\beta_i \sum_{j=1}^n\gamma_je_j^{q^i}
\end{align*}
Defining $z=-\sum_{j=1}^n\gamma_je_j^{q^r}$ and $w_i=\sum_{j=1}^n\gamma_je_j^{q^i}$, proving the equation is equivalent to prove that $$z=\langle\bm{\beta},\bm{w}\rangle.$$ The latter equation will be checked thanks to the multiplication protocol from \cite{BN20} (adapted in the matrix setting, see~\cite{F22}).
In the case of MinRank, we want to apply this protocol to $\bm{E}$, which we will have to compute. In what follows, we work on the columns of $\bm{E}$ and in $\Fqm$. Depending on the parameters, it might also be easier to work on the rows of $\bm{E}$.
This gives us the following protocol from \cite{F22}: 
\begin{figure}[H]
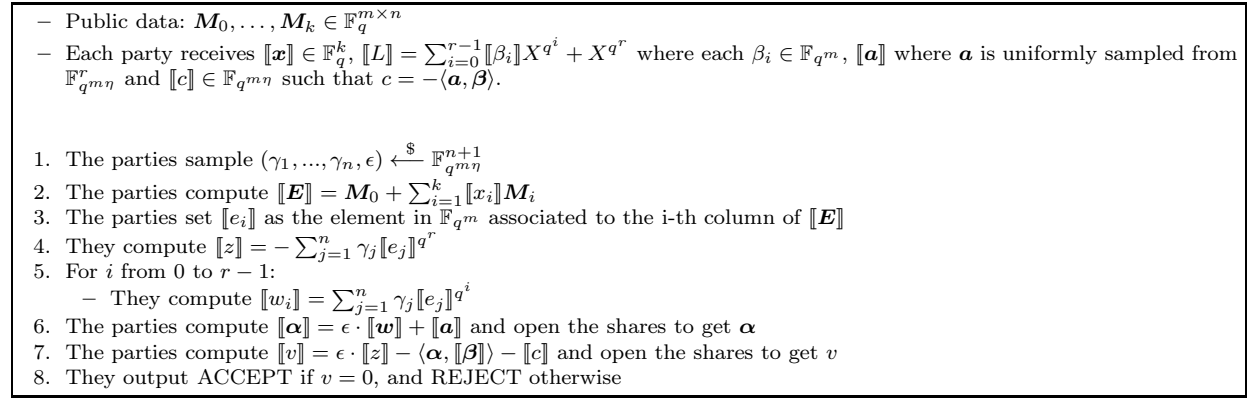

\pcb[codesize=\scriptsize, minlineheight=0.75\baselineskip, mode=text, width=0.98\textwidth] {
  \begin{itemize}
    \item Public data: ${\bm{M}_0},\dots,{\bm{M}_k} \in \Fq^{m\times n}$
    \item Each party receives $[\![\bm{x}]\!] \in \mathbb{F}_q^k$, $[\![L]\!]=\sum_{i=0}^{r-1}[\![\beta_i]\!]X^{q^i}+X^{q^r}$ where each $\beta_i\in\Fqm$, $[\![\bm{a}]\!]$ where $\bm{a}$ is uniformly sampled from $\mathbb{F}_{q^{m\eta}}^r$ and $[\![c]\!]\in\mathbb{F}_{q^{m\eta}}$ such that $c=-\langle \bm{a},\bm{\beta}\rangle$.
  \end{itemize}
  \begin{enumerate}
    \item The parties sample $(\gamma_1,...,\gamma_n,\epsilon) \sampler \mathbb{F}_{q^{m\eta}}^{n+1}$
    \item The parties compute $\share{\bm{E}} = \bm{M}_0 + \sum_{i=1}^{k}[\![x_i]\!]\bm{M}_i$
    \item The parties set  $[\![e_i]\!]$ as the element in $\Fqm$ associated to the i-th column of $[\![\bm{E}]\!]$
    \item They compute $[\![z]\!] = -\sum_{j=1}^n \gamma_j[\![e_j]\!]^{q^r}$
    \item For $i$ from $0$ to $r-1$:
    \begin{itemize}
      \item They compute $[\![w_i]\!]=\sum_{j=1}^n \gamma_j[\![e_j]\!]^{q^i}$
    \end{itemize}
    \item The parties compute $[\![\bm{\alpha}]\!] = \epsilon \cdot [\![\bm{w}]\!]+[\![\bm{a}]\!]$ and open the shares to get $\bm{\alpha}$
    \item The parties compute $[\![v]\!]=\epsilon \cdot [\![z]\!]-\langle\bm{\alpha},[\![\bm{\beta}]\!]\rangle-[\![c]\!]$ and open the shares to get $v$
    \item They output ACCEPT if $v=0$, and REJECT otherwise
  \end{enumerate}
}
\vspace{-\baselineskip}
\captionof{figure}{Protocol $\Pi^{\eta}$ to check that an input is solution of an instance of MinRank}
\label{pi_eta}
\end{figure}
One should stress that we can locally compute shares to the power $q^i$. Since we are in $\Fq$, this operation is linear, and thus this operation is basically ``free'' in multiparty computation.

\begin{proposition}[\cite{F22}] If $\operatorname{W}_R \:(\bm{E})\leq r$, then the protocol $\Pi^\eta$ always accepts. If $\operatorname{W}_R \:(\bm{E})> r$, the protocol accepts with  probability at most: $p_{\Pi,\eta}=\frac{2}{q^{m\eta}}-\frac{1}{q^{2m\eta}}$ (we call $p_{\Pi,\eta}$ the false positive rate of $\Pi^{\eta}$).
\end{proposition}

\begin{proof}
The proof is the same as the one in \cite[Section 5.1]{F22}
\end{proof}

Thanks to this MPC protocol, we can build a zero-knowledge proof of knowledge protocol, using either additive secret sharing (Fig.\ref{hypercubempcith}) or threshold secret sharing (Fig.\ref{thmpcith}).

%% file: 3-0-Proof_MinRank.tex
\subsection{Description of MIRA-Additive}
\input{3-1-zk_proof_hyp}

\subsection{Description of MIRA-Threshold}
\input{3-2-zk_proof_th}

%% file: 3-1-zk_proof_hyp.tex
%!TEX root = ./main.tex

\subsubsection{Proof of Knowledge with Additive Secret Sharing Scheme} \hfill

We describe here the idea of the hypercube technique introduced by \cite{AGHHJY22}.
\begin{itemize}
    \item The prover generates $N = 2^D$ shares of the inputs of the $\Pi^\eta$ protocol (Fig.\ref{pi_eta}), using an additive secret sharing scheme;
    \item He then computes the shares associated to the ``main parties'', i.e, the parties associated to each dimension, by summing up the $2^{D-1}$ shares that have the same index on the dimension (in the same way as explained in section \ref{mpc_pok});
    \item There are $D$ dimensions, each of them having $2$ main parties. The prover sends the commitments of all the $N$ shares; 
    \item He then receives the first challenge, which is some random values for the computation of the $\Pi^\eta$ protocol;
    \item Then, for $D+1$ main parties, he executes the protocol $\Pi^\eta$. For the $D-1$ other parties, he will use the broadcasted values of $\bm{\alpha}$ and $v$ in order to compute their shares of $\bm{\alpha}$ and of $v$. He will hash these values;

     \item He sends to the verifier the hash of these values (it is possible to hash them twice, as is done in the description of Fig.\ref{exec_pi_hypercube}, or not); 
    \item He receives the second challenge, which is a single leaf, $i^*$;
    \item He reveals every commitment except the $i^*$\textsuperscript{th} one, and sends the share $\share{\bm{\alpha}}_{i^*}$ broadcast by the party $i^*$;
    \item The verifier can then reconstruct the broadcast shares of the main parties, i.e the $2$ shares for each dimension, and then check that everything is correct. Similarly to the prover, he can avoid the computation of the protocol for $D-1$ parties.
\end{itemize}

One detail to notice is that in Fig.\ref{exec_pi_hypercube}, we use an hash function to output the broadcast shares. We stress that this hash function is not mandatory. The security of the protocol with or without it is the same. Depending on the implementation, using it can be more convenient (for example, to save memory when the used implementation of the hash functions is not incremental). In our security proofs, we will consider that such hash function is used. 
%However, we emphasize on the fact that the security is equivalent in both cases (as long as the used hash function is secure).

Finally, if $i^* \ne N$, the prover has to send $\share{\bm{x}}_N, \share{\bm{\beta}}_N,\share{c}_N$ in addition to the sibling path. This comes from the fact that these shares are not completely random.

This gives us the protocol provided in Fig.\ref{hypercubempcith}.
We insist on the fact that, as mentioned in section \ref{mpc_pok}, a prover doesn't need to simulate $2D$ parties in step 10, but only $D+1$, as the values of $\bm{\alpha}$ and $v$ are the same, no matter the dimension. In fact, for $D-1$ dimensions in step 10, it is possible to retrieve $\share{\bm{\alpha}}_{(k,2)}$ by computing $\bm{\alpha}-\share{\bm{\alpha}}_{(k,1)}$ instead of simulating the MPC protocol. We emphasize on this optimisation, as this is a crucial advantage of the hypercube structure.

\input{fig_pok_hypercube}
\input{fig-mpc_exec_hypercube}
\input{fig-mpc_check_hypercube}

We prove that our scheme satisfies the necessary security properties in the proof of theorem \ref{Theoremhypercube} in section \ref{sec_proofs_pok}.

\subsubsection{MIRA-Additive} \hfill

From the zero-knowledge protocol, we can simply deduce a signature scheme, using the Fiat-Shamir transform. The protocol is now non-interactive, and we deterministically sample the challenge thanks to a hash function. We must repeat the zero-knowledge protocol several times, in order to reach a certain level of security. Note that we use a value $\salt$, as in \cite{FJR22}, in order to increase the security of the scheme as it reduces the probability to have seeds collisions. Moreover, simarly than in the proof of knowledge, the hash function $\hash_3$ is not mandatory and depends only of implementation choices. The signature protocol obtained is described in Fig.\ref{add_sig} and Fig.\ref{add_sig_verify}.
As aforementioned in the proof of knowledge, the signature process avoid the computation of the MPC protocol for $D-1$ parties in step 8. The verification process avoids it in step 6 as well.

\input{fig-sign_hypercube}
\input{fig-verify_hypercube}

We prove that the scheme is EUF-CMA secure in the proof of theorem \ref{secu_sig_hypercube} in section \ref{sec_proofs_sig}.

\clearpage

%% file: fig_pok_hypercube.tex
%!TEX root = ./main.tex

\begin{figure}[ht]
\pcb[codesize=\scriptsize, minlineheight=0.75\baselineskip, mode=text, width=0.98\textwidth] { 
  \pcind - Public data $\bm{M}_0,\dots,\bm{M}_k \in \Fq^{m\times n}$\\ 
  The prover wants to convince the verifier that he knows a solution $\bm{x}\in \Fqk$ of the MinRank instance, i.e  such as $\bm{E} = \bm{M}_0 + \sum_{i=1}^k(\bm{M}_ix_i) $ and $\operatorname{W}_R(\bm{E}) \le r$ \\[\baselineskip]
  \textbf{Step 1: Commitment} \\
 1. The prover sets $U = \langle e_1 \dots e_n \rangle$, and computes $L(X) = \prod_{u\in U} (X-u) = X^{q^r}+\sum_{i=0}^{r-1} \beta_i \cdot X^{q^i}$ for some $\beta_i \in \Fqm^{r}$\\
  2. The prover samples a root seed: $\seed\sampler\{0, 1\}^\lambda$\\
  3. The prover expands the root seed recursively using TreePRG to obtain $N$ leaves, from which he derives $N$ seeds and $N$ commitment random tapes $(\seed_{i'},\rho_{i'})$ \\
  4. For each $i\in\oneto{N-1}$: \\
  \pcind \pcind - Sample $(\share{\bm{x}}_i,\share{\bm{\beta}}_i,\share{\bm{a}}_i\share{c}_i)\samples{\seed_i} PRG$ where PRG is a pseudo-random generator\\
  \pcind \pcind - $\state_i=\seed_i$\\
  5. For the share $N$: \\
  \pcind \pcind - Sample $\share{\bm{a}}_{N}\samples{\seed_{N}} PRG$\\
  \pcind \pcind - Compute $\share{\bm{x}}_{N}=\bm{x}-\sum_{i=1}^{N-1}\share{\bm{x}}_i$, $\share{\bm{\beta}}_{N}=\bm{\beta}-\sum_{i=1}^{N-1}\share{\bm{\beta}}_i$ and $\share{c}_{N}=-\dotp{\boldsymbol{a}, \bm{\beta}}-\sum_{i=1}^{N-1}\share{c}_i$\\
  \pcind \pcind - $\state_{N}=(\seed_{N},\share{\bm{x}}_{N},\share{\bm{\beta}}_{N},\share{c}_{N})$\\
  6. The prover computes the commitments: $\cmt_i=\commit{\state_i,\rho_i}$, for each $i\in\oneto{N}$.\\
  7. The prover computes and sends $h_0 = \hash(\cmt_1, \cdots, \cmt_{N})$ to the verifier.\\
  8. The prover computes the input shares of the main parties by summing all the associated leaves. We index each party by its coordinates on the hypercube: $i=(i_1,...,i_D)$ where $i_k\in\oneto{2}$. For all main party index $p\in\lbrace (1,1),...,(D,2)\rbrace$: $\share{\bm{x}}_{(p_1,p_2)}=\sum_{i: i_{p_1}=p_2}\share{\bm{x}}_i$, $\share{\bm{\beta}}_{(p_1,p_2)}=\sum_{i: i_{p_1}=p_2}\share{\bm{\beta}}_i$, $\share{\bm{a}}_{(p_1,p_2)}=\sum_{i: i_{p_1}=p_2}\share{\bm{a}}_i$ and $\share{c}_{(p_1,p_2)}=\sum_{i: i_{p_1}=p_2}\share{c}_i$.\\
  [\baselineskip]
  \textbf{Step 2: First Challenge} \\
  9. The verifier sends $\big( (\gamma_j)_{j \in \oneto{n}}, \epsilon \big) \in \Fqme^n \times \Fqme$ to the prover.\\[\baselineskip]
  \textbf{Step 3: First Response} \\
  10. For each dimension $k\in\oneto{D}$, the prover executes the algorithm in Fig. \ref{exec_pi_hypercube} on the set of the main parties of the dimension. This set is denoted $P_k$. He computes $\big((\share{\bm{\alpha}}_{(k,i)},\share{v}_{(k,i)})_{i \in \oneto{2}},H_k\big)\longleftarrow \text{Algorithm Fig.\ref{exec_pi_hypercube}}(\big( P_k,(\gamma_j)_{j \in \oneto{n}}, \epsilon \big)$. He only needs to execute the algorithm $D+1$ times, as explained above.\\
11. The prover commits the executions: $h_1=\mathsf{H}(H_1,...,H_D)$.\\[\baselineskip]
  \textbf{Step 4: Second Challenge}\\
  12. The verifier gets a random leaf $i^*\in\oneto{N}$ and sends it to the prover.\\[\baselineskip]
  \textbf{Step 5: Second Response and Verification}\\
  13. The prover sends to the verifier: $\left(\cmt_{i^*},\share{\bm{\alpha}}_{i^*},(\state_j,\rho_j)_{j\neq i^*}\right)$, where $(\state_j,\rho_j)_{j\neq i^*}$ corresponds to the sibling path to this commitment. If $i^* \ne N$, he also has to send $\share{\bm{x}}_{N},\share{\bm{\beta}}_{N}$ and $\share{c}_{N}$.\\
  14. The verifier can deduce all the leaves (with the exception of index $i^*$), and recover $h_0$ using the sibling path and $\cmt_{i^*}$.\\
  15. For all dimension $k\in\oneto{D}$, the verifier runs the algorithm in Fig. \ref{check_pi_hypercube} to get $(\share{\bm{\alpha}}_{(k,i)},\share{v}_{(k,i)})_{i \in \oneto{2}}$ and $H_k$. Then he checks that\\
  \pcind \pcind - $v=0$ for all dimensions.\\
  \pcind \pcind - $\bm{\alpha}$ is the same for all dimensions.\\
  \pcind \pcind -  $h_1=\mathsf{H}(H_1,...,H_D)$
  }
\vspace{-\baselineskip}
\caption{\footnotesize{MinRank Proof of Knowledge with additive-sharing MPCitH}}\label{hypercubempcith}

\end{figure}

%% file: fig-mpc_exec_hypercube.tex
%!TEX root = ./main.tex

\begin{figure}[ht]
\pcb[codesize=\scriptsize, minlineheight=0.75\baselineskip, mode=text, width=0.98\textwidth] { 
  \textbf{Inputs:} A set of main shares for the dimension $k$: $\big((\share{\bm{x}}_{(k,i)},\share{\bm{\beta}}_{(k,i)},\share{\bm{a}}_{(k,i)}, \share{c}_{(k,i)})\big)_{i \in \oneto{2}}$ and a protocol challenge $\big( (\gamma_j)_{j \in \oneto{n}}, \epsilon \big)$\\
   \textbf{Outputs:} A set of main shares $\big( \bm{\alpha}_{(k,i)},\share{v}_{(k,i)}\big)_{i \in \oneto{2}}$ and a commitment $H_k$ of the execution\\[\baselineskip]
  \textbf{For} each main party $i\in\{1,2\}$: \\
\pcind - Compute $\share{\bm{E}}_{(k,i)} = \bm{M}_0 + \sum_{j=1}^{k}[\![x_j]\!]_{(k,i)}\bm{M}_j$ \\ 
  \pcind - Set  $[\![e_j]\!]_{(k,i)}$ as the element in $\Fqm$ associated to the j-th column of $[\![\bm{E}]\!]_{(k,i)}$ \\
  \pcind - Compute $[\![z]\!]_{(k,i)} = -\sum_{j=1}^n \gamma_j\cdot[\![e_j]\!]_{(k,i)}^{q^r}$ \\
 \pcind \textbf{For} $l$ from $0$ to $r-1$: \\
\pcind  \pcind $\diamond$ Compute $[\![w_l]\!]_{(k,i)}=\sum_{j=1}^n \gamma_j\cdot[\![e_j]\!]_{(k,i)}^{q^l}$ \\
  \pcind - Compute $[\![\bm{\alpha}]\!]_{(k,i)} = \epsilon \cdot [\![\bm{w}]\!]_{(k,i)}+[\![\bm{a}]\!]_{(k,i)}$ and reveal $\bm{\alpha}$ \\
  \pcind - Compute $[\![v]\!]_{(k,i)}=\epsilon \cdot [\![z]\!]_{(k,i)}-\langle\bm{\alpha},[\![\bm{\beta}]\!]_{(k,i)}\rangle-[\![c]\!]_{(k,i)}$ \\
 Compute $H_k=\hash\big((\share{\bm{\alpha}}_{(k,i)},\share{v}_{(k,i)})_{i \in \oneto{2}}\big)$
  }
\vspace{-\baselineskip}
\captionof{figure}{\footnotesize{Execution of the MPC protocol $\Pi^\eta$ on a set of main shares}}
\label{exec_pi_hypercube}
 \end{figure}

%% file: fig-mpc_check_hypercube.tex
\begin{figure}[ht]
\pcb[codesize=\scriptsize, minlineheight=0.75\baselineskip, mode=text, width=0.98\textwidth] {
  \textbf{Inputs:} A leaf $i^*$ that one does not reveal. $\Tilde{i}$ is the index depending on the leaf $i^*$. The share $\share{\bm{\alpha}}_{i^*}$, and all the main parties shares $\big((\share{\bm{x}}_{(k,i)},\share{\bm{\beta}}_{(k,i)},\share{\bm{a}}_{(k,i)},\share{c}_{(k,i)})\big)_{i \in \oneto{2}}$. For $i = \Tilde{i}$, the main shares correspond to the usual main party shares, except the share $i^*$ is missing. \\
   \textbf{Outputs:}  $(\share{\bm{\alpha}}_{(k,i)},\share{v}_{(k,i)})_{i \in \oneto{2}}$ and a commitment $H_k$ of the execution\\[\baselineskip]
  \textbf{For} $i \in \oneto{2}$:\\
  \pcind \pcind \pcind \textbf{If} $i=\Tilde{i}$:\\
  \pcind \pcind \pcind \pcind Set $\share{v}_{i^*}$ such that $v=0$ and compute $\share{v}_{(k,\Tilde{i})}$\\
  \pcind \pcind \pcind \pcind Compute $\share{\bm{\alpha}}_{(k,\Tilde{i})} =\sum_{j:j_k = \Tilde{i},j \ne i^*}\share{\bm{\alpha}}_{j} +\share{\bm{\alpha}}_{i^*}$ where $\sum_{j:j_k = \Tilde{i},j \ne i^*}\share{\bm{\alpha}}_{j}$ is obtained by applying $\Pi^\eta$ to the partially-aggregated input shares.\\
  \pcind \pcind \textbf{Else}:\\
  \pcind \pcind \pcind \pcind Do the same computations as in $\Pi^\eta$ to obtain the correct shares $\share{\bm{\alpha}}_{(k,i)}$ and $\share{v}_{(k,i)}$\\
  Compute $H_k=\hash\big((\share{\bm{\alpha}}_{(k,i)},\share{v}_{(k,i)})_{i \in \oneto{2}}\big)$
  }
\vspace{-\baselineskip}
\captionof{figure}{\footnotesize{Check the executions of the MPC protocol $\Pi^\eta$ on the main parties}}\label{check_pi_hypercube}
\end{figure}

%% file: fig-sign_hypercube.tex
\begin{figure}[ht]
\pcb[codesize=\scriptsize, minlineheight=0.75\baselineskip, mode=text, width=0.98\textwidth] { 
\textbf{Inputs}\\
  \pcind - Secret key $\sk = (\bm{x}) \in \mathbb{F}_q^k$, $(\bm{E}) \in \Fq^{m\times n}$ such that $\operatorname{W}_R(\bm{E}) \leq r$\\
  \pcind - Public data $\bm{M}_0,\dots,\bm{M}_k \in \Fq^{m\times n}$ with $\bm{E} = \bm{M}_0 + \sum_{i=1}^k\bm{M}_ix_i $\\
  \pcind - Message $m \in \{0, 1\}^*$ \\[\baselineskip]
  \textbf{Step 1: Commitment} \\
  1. Sample a random salt value $\salt \sampler \{0, 1\}^{2\lambda}$ \\
  2. Set $U = \langle e_1, \dots, e_n \rangle $ and compute $\bm{\beta} = (\beta_k)_{k \in [0, r - 1]} \in \Fqm^r$ the coefficients of the annihilator q-polynomial $L(X)$ associated to $\bm{{E}}$ such that $L(X) = \prod_{u \in U} (X-u) = X^{q^r} + \sum_{k = 0}^{r - 1} \beta_k \cdot X^{q^k}$ and $\forall j \in \oneto{n}, L({{e}}_j) = 0$ where each $e_j$ is the element in $\Fqm$ associated to the j-th column of $\bm{E}$  \\
  3. For each iteration $e \in \oneto{\tau}$: \\
  \pcind - Sample a root seed: $\seed^{(e)}\leftarrow \lbrace 0,1\rbrace^\lambda$ \\
  \pcind - Expand root seed recursively with TreePRG to obtain $N$ seeds and randomness $(\seed^{(e)}_{i'},\rho_{i'})$ \\
  \pcind \textbf{For} each $i\in\oneto{N}$: \\
  \pcind \pcind \pcind \textbf{If} $i\neq N$: \\
  \pcind \pcind \pcind \pcind \pcind $\diamond$ Sample $(\share{\bm{x}}^{(e)}_i,\share{\bm{\beta}}^{(e)}_i,\share{\bm{a}}^{(e)}_i\share{c}^{(e)}_i)\samples{\seed_i^{(e)}} PRG$ where PRG is a pseudo-random generator \\
  \pcind \pcind \pcind \pcind \pcind $\diamond$ $\state^{(e)}_i=\seed^{(e)}_i$\\
  \pcind \pcind \pcind \textbf{Else}: \\
  \pcind \pcind \pcind \pcind \pcind $\diamond$ Sample $\share{\bm{a}}^{(e)}_{N}\samples{\seed^{(e)}_{N}} PRG$\\
  \pcind \pcind \pcind \pcind \pcind $\diamond$ Compute $\share{\bm{x}}^{(e)}_{N}=\bm{x}-\sum_{i=1}^{N-1}\share{\bm{x}}^{(e)}_i$, $\share{\bm{\beta}}^{(e)}_{N}=\bm{\beta}-\sum_{i=1}^{N-1}\share{\bm{\beta}}^{(e)}_i$ and $\share{c}^{(e)}_{N}=-\dotp{\boldsymbol{a}, \bm{\beta}}-\sum_{i=1}^{N-1}\share{c}^{(e)}_i$\\
  \pcind \pcind \pcind \pcind \pcind $\diamond$ $\state_{N}^{(e)}=(\seed_{N}^{(e)},\share{\bm{x}}^{(e)}_{N},\share{\bm{\beta}}^{(e)}_{N},\share{c}^{(e)}_{N})$\\
  \pcind \pcind \pcind - Compute $\cmt_i^{(e)}=\mathsf{H}_0(\salt,e,i,\state_i^{(e)})$\\
  \pcind 4. For all main party index $p=(p_1,p_2) \in\lbrace (1,1),...(D,2)\rbrace$: $\share{\bm{x}}_{(p_1,p_2)}=\sum_{i: i_{p_1}=p_2}\share{\bm{x}}_i$, $\share{\bm{\beta}}_{(p_1,p_2)}=\sum_{i: i_{p_1}=p_2}\share{\bm{\beta}}_i$, $\share{\bm{a}}_{(p_1,p_2)}=\sum_{i: i_{p_1}=p_2}\share{\bm{a}}_i$ and $\share{c}_{(p_1,p_2)}=\sum_{i: i_{p_1}=p_2}\share{c}_i$.\\
  \pcind 5. Compute and commit $h_0^{(e)} =\hash_1(\salt,e,\cmt_1^{(e)},...,\cmt_{N}^{(e)})$\\
  \pcind 6. Compute $h_1=\mathsf{H}_2\left(\salt,m,h_0^{(1)},...,h_0^{(\tau)}\right)$\\
  [\baselineskip]
  \textbf{Step 2: First Challenge} \\
  7. Sample $\big( (\gamma^{(e)}_j)_{j \in \oneto{n}}, \epsilon^{(e)} \big)_{e \in \oneto{\tau}} \sampler PRG(h_1)$ where $ \big( (\gamma^{(e)}_j)_{j \in \oneto{n}}, \epsilon^{(e)} \big)_{e \in \oneto{\tau}} \in (\Fqme^n \times \Fqme)^{\tau}$ \\
  [\baselineskip]
  \textbf{Step 3: First Response} \\
  8. For each iteration $e \in \oneto{\tau}$: \\
  \pcind \textbf{For} each each dimension $k\in\oneto{D}$: \\
  \pcind \pcind - Execute the algorithm in Fig.\ref{exec_pi_hypercube} to obtain $\big(\share{\bm{\alpha}^{(e)}}_{(k,i)}, \share{v^{(e)}}_{(k,i)}\big)_{i \in \oneto{2}}$ and $H_k^{(e)}$, with the hash function $\hash_3$ \\
  9. Compute $h_2 = \hash_4(m, \pk, \salt, h_1, (H_1^{(e)} \dots H_D^{(e)})_{e \in \oneto{\tau}})$\\
  [\baselineskip]
  \textbf{Step 4: Second Challenge} \\
  10. Sample $i^{*(e)}\sampler PRG(h_2) $ where $ i^{*(e)} \in [1,N]^{\tau}$ \\
  [\baselineskip]
  \textbf{Step 5: Second Response} \\
  11.For each iteration $e \in \oneto{\tau}$: \\
  \pcind - Compute $\rsp^{(e)}=\left( (\state_j)^{(e)}_{j\neq i^*},\cmt_{i^{*(e)}},\share{\bm{\alpha}^{(e)}}_{i^{*(e)}}\right)$ where $(\state_j)^{(e)}_{j\neq i^*}$ corresponds to the sibling path. If $i^* \ne N$, he also has to send $\share{\bm{x}}_{N},\share{\bm{\beta}}_{N}$ and $\share{c}_{N}$. \\
  12. Output $\sigma = \left(\salt, h_1, h_2, (\rsp^{(e)})_{e \in \oneto{\tau}}\right)$
}
\vspace{-\baselineskip}
\captionof{figure}{\footnotesize{MIRA Signature Scheme based on additive sharing MPCitH - Signature Algorithm}}
\label{add_sig}
\end{figure}

%% file: fig-verify_hypercube.tex
%!TEX root = ./main.tex

\begin{figure}[ht]
\pcb[codesize=\scriptsize, minlineheight=0.75\baselineskip, mode=text, width=0.98\textwidth] { 
\textbf{Inputs}\\
  \pcind - Public data $\bm{M}_0,\dots,\bm{M}_k \in \Fq^{m\times n}$ such that there is $\bm{x} \in \Fq^k$ such that $\bm{E} = \bm{M}_0 + \sum_{i=1}^k(\bm{M}_ix_i) $ and $\operatorname{W}_R(\bm{E}) \leq r$\\
  \pcind - Message $m \in \{0, 1\}^*$ \\
  \pcind - Signature $\sigma = (\salt, h_1, h_2, (\bar{\rsp}^{(e)})_{e \in \oneto{\tau}})$ \\[\baselineskip]
  \textbf{Step 1: Parse signature} \\
  1. Sample $\big( (\gamma^{(e)}_j)_{j \in \oneto{n}}, \epsilon^{(e)} \big)_{e \in \oneto{\tau}} \sampler PRG(h_1)$ where $ \big( (\gamma^{(e)}_j)_{j \in \oneto{n}}, \epsilon^{(e)} \big)_{e \in \oneto{\tau}} \in (\Fqme^n \times \Fqme)^{\tau}$ \\
  2. Sample $i^{*(e)}\sampler PRG(h_2) $ where $ i^{*(e)}  \in [1,N]$\\
  3. Parse $\left( (\state_j)^{(e)}_{j\neq i^*},\cmt_{i^{*(e)}},\share{\bm{\alpha}^{(e)}}_{i^{*(e)}}\right) \gets \rsp^{(e)}$\\
   %$(\state_i^{(e)})_{i\in \oneto{N}}$ in $\sigma$ for each iteration $e \in \oneto{\tau}$\\
  [\baselineskip]
  \textbf{Step 2: Recompute $h_1$}\\
  4. \textbf{For} each iteration $e \in \oneto{\tau}$: \\
  \pcind - \textbf{For} each $i\neq i^{*(e)}$:\\
  \pcind \pcind $\diamond$ $\cmt_{i}^{(e)}=\mathsf{H}_0(\salt,e,i,\state_{i}^{(e)})$\\
  \pcind - Compute $h_0^{(e)} = \hash_1\left(\salt,e,\cmt_1^{(e)},...,\cmt_{N}^{(e)}\right)$\\
  5. Compute $\bar{h}_1=\hash_2\left(\salt,m,h_0^{(1)},...,h_0^{(\tau)}\right)$\\
  [\baselineskip]
  \textbf{Step 3: Recompute $h_2$}\\
  6. \textbf{For} each iteration $e \in \oneto{\tau}$: \\
  \pcind - Simulate MPC protocol $\Pi^\eta$ on main parties. \\
  \pcind - \textbf{For} each dimension $k\in\oneto{D}$:\\
  \pcind \pcind $\diamond$ Run the algorithm in Fig. \ref{check_pi_hypercube} to get $\bar{H_k^{(e)}}$, in which $\share{\bm{\alpha}^{(e)}}_{i^{*(e)}}$ is used to compute the share of $\bm{\alpha}$ of the main party relying on $i^*$.\\
  7. Compute $\bar{h}_2 = \hash_4\left(m, \pk, \salt, \bar{h}_1, \left(\bar{H_1^{(e)}},...,\bar{H_D^{(e)}}\right)_{e \in \oneto{\tau}}\right)$  \\[\baselineskip]
  \textbf{Step 4: Verify signature} \\
  8. Return $(\bar{h}_1 = h_1) \wedge (\bar{h}_2 = h_2)$
  \\
}
\vspace{-\baselineskip}
\captionof{figure}{\footnotesize{MIRA Signature Scheme based on additive MPCitH - Verification Algorithm}}
\label{add_sig_verify}
\end{figure}

%% file: 3-2-zk_proof_th.tex
%!TEX root = ./main.tex

\subsubsection{Proof of Knowledge with Threshold Secret Sharing} \hfill \\

Once we have this zero-knowledge protocol (and signature scheme) with additive sharing, one may wonder if it is possible to use other secret sharing schemes to build MPCitH-based signature schemes. The answer is positive and given in \cite{FR22}: we can use any threshold linear secret sharing scheme. (Ligero~\cite{ligero} already considered Shamir's secret sharing but this scheme is only interesting for ``medium-size circuits'', which is not the case of circuits involved to build signature schemes.) When dealing with threshold secret sharing, we will consider (except when told otherwise) that we are referring to Shamir's secret sharing scheme. As reminded earlier (see Definition~\ref{def:shamir-sharing}), in order to build a sharing of a secret in this setting, one has to build the random polynomial with the secret as the constant term, and then evaluate this polynomial in the evaluation points. Since most of our secrets in our case are vectors, this has to be done for every coordinates, and so the notation can be pretty heavy. For the sake of simplicity, this process is not made explicit in the description of our protocols.

The protocol with threshold secret sharing works in the following way:
\begin{itemize}
    \item Using a $(\ell+1,N)$-threshold secret sharing scheme, the prover computes the shares of the inputs of the protocol $\Pi^\eta$, computes the commitments, and then computes the Merkle tree root of the commitments;
    \item He then receives the first challenge, i.e, the randomness for the protocol $\Pi^\eta$;
    \item He then chooses a set $S$ of $\ell+1$ parties, for which he executes the protocol $\Pi^\eta$, and computes the hash value of the executions;
    \item He then receives the second challenge. This challenge is a set $I$ of $\ell$ parties (instead of just $1$ as in the additive protocol). For these $\ell$ parties, he sends the authentication path of the commitments, and the elements to rebuild their shares. He then chooses one party, $i^*$, in $S \setminus I$, and reveal the value $\share{\bm{\alpha}}_{i^*}$ computed for this party;
    \item Thanks to this value, the verifier can recover everything he needs, and can then authenticate, or not, the prover.
\end{itemize}
Note that the knowledge of $\ell$ shares provides no information on the secret, since we use a $(\ell+1,N)$-threshold secret sharing. Moreover, this method permits to reduce the number of operations, for the cost a larger signature size (we refer to \cite{FR22} for more details).

However, one limitation arises: when using Shamir's secret sharings (or any low-threshold LSSS), we need to have $N \leq q$ to share values of $\Fq$. 
Finally, this gives us the proof of knowledge of Fig.\ref{thmpcith}.

\input{fig_pok_th}
\input{fig-exec_pi_th}

Let us look at the soundness of the protocol.
If the probability of a false positive of the MPC protocol was equal to $0$ (i.e, the probability of a false positive of $\Pi^\eta$ was $0$), then a malicious prover (who doesn't know a solution of the MinRank instance) would need to cheat for exactly $N-\ell$ parties. Indeed, if he cheated for less parties, then at least $\ell+1$ shares would be consistent. Since we suppose the probability of false-positive is $0$, this means he would know a good witness. If he cheated for more than $N-\ell$ parties, the verifier would always ask for the opening of a corrupted party, hence he would always discover the cheat. This means that the malicious prover needs to cheat on exactly $N-\ell$ parties, and the only way that the verifier is convinced is when the verifier asks the opening of the exact set of the $\ell$ honest parties. This means that the probability that the verifier does not detect the cheat is $\frac{1}{\binom{N}{\ell}}$. \\
However, $p$ is not $0$ in our case. We can make the same reasoning, and furthermore, consider this false positive rate. We would expect to get a probability for a malicious prover to convince a verifier of $\frac{1}{\binom{N}{\ell}}+(1-\frac{1}{\binom{N}{\ell}})\cdot p$ (We remind the reader that $p =\frac{2}{q^{m\eta}}-\frac{1}{q^{2m\eta}} $). \\
Unfortunately, it is just a lower bound of the soundness of the protocol. As remarked in \cite{FR22}, there is a more effective way for a malicious prover to cheat at the protocol, as we will see in the proof of the theorem \ref{Theoremth } in section \ref{sec_proofs_pok}. \\[0.3cm]

\subsubsection{MIRA-Threshold} \hfill \\

From the zero-knowledge protocol, we can deduce a signature scheme thanks to the Fiat-Shamir transformation. The resulting scheme is non-interactive and we deterministically sample the challenge using a hash function. We must repeat the zero-knowledge protocol several times ($\tau$ times), in order to achieve a certain level of security. The obtained signature scheme is described in Fig.\ref{th_sig} and Fig.\ref{th_sig_verify}.
\input{fig-sign}
\input{fig-verify}

We prove that the scheme is EUF-CMA secure in the proof of theorem \ref{secu_sig_th} in \ref{sec_proofs_sig}.

\clearpage

%% file: fig_pok_th.tex
\begin{figure}[ht]
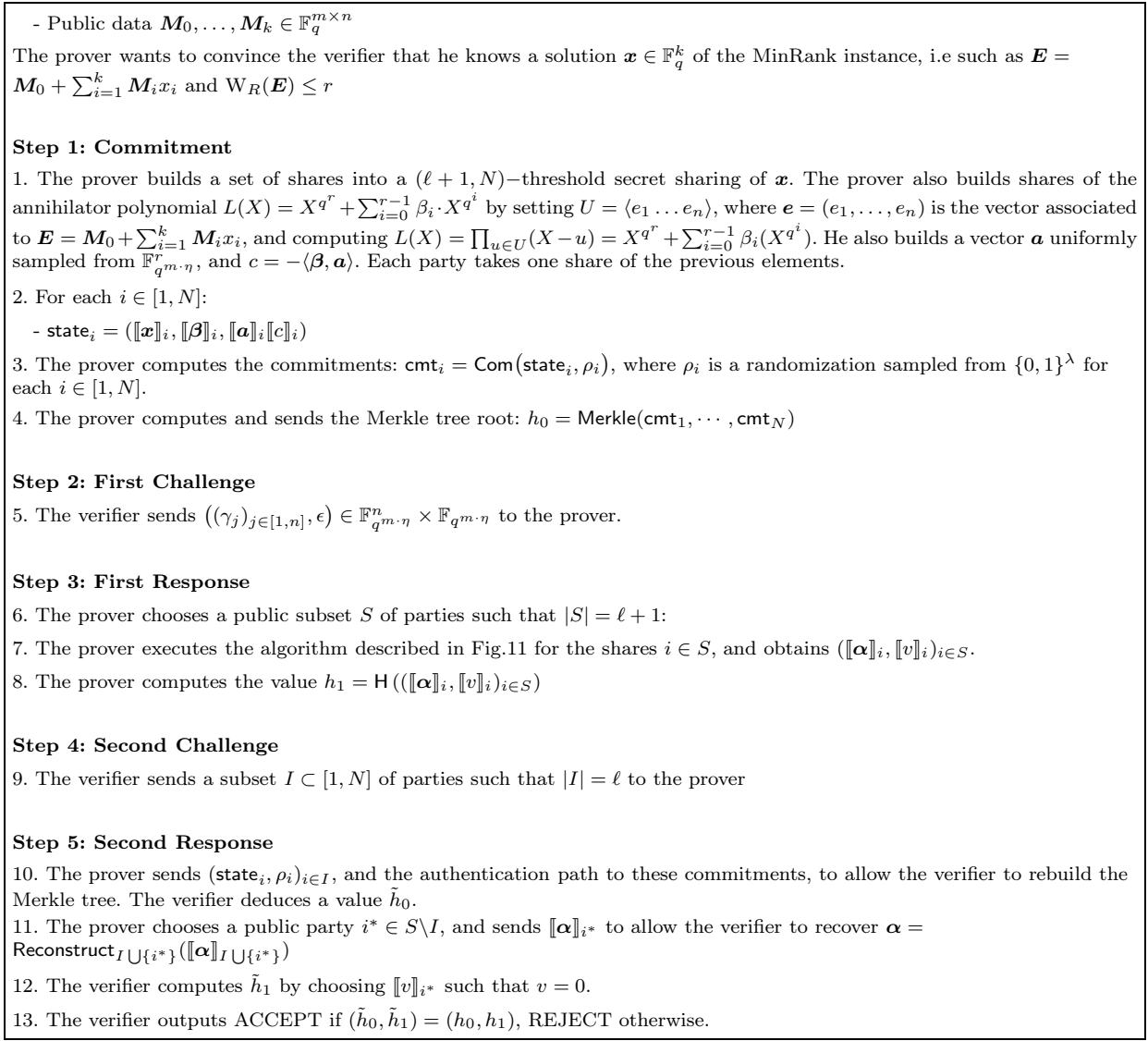

\pcb[codesize=\scriptsize, minlineheight=0.75\baselineskip, mode=text, width=0.98\textwidth] { 
  \pcind - Public data $\bm{M}_0,\dots,\bm{M}_k \in \Fq^{m\times n}$\\ 
  The prover wants to convince the verifier that he knows a solution $\bm{x}\in \Fqk$ of the MinRank instance, i.e  such as $\bm{E} = \bm{M}_0 + \sum_{i=1}^k\bm{M}_ix_i $ and $\operatorname{W}_R(\bm{E}) \le r$ \\[\baselineskip]
  \textbf{Step 1: Commitment} \\
  1. The prover builds a set of shares into a $(\ell+1,N)-$threshold secret sharing of $\bm{x}$. The prover also builds shares of the annihilator polynomial $L(X)=X^{q^r}+\sum_{i=0}^{r-1}\beta_i\cdot X^{q^i}$ by setting $U = \langle e_1 \dots e_n \rangle$, where $\bm{e} = (e_1, \dots, e_n)$ is the vector associated to $\bm{E} = \bm{M}_0 + \sum_{i=1}^k\bm{M}_ix_i $, and computing $L(X) = \prod_{u\in U} (X-u) = X^{q^r}+\sum_{i=0}^{r-1} \beta_i(X^{q^i})$. He also builds a vector $\bm{a}$ uniformly sampled from $\Fqme^{r}$, and $c=-\langle\bm{\beta},\bm{a}\rangle$. Each party takes one share of the previous elements.\\
  2. For each $i\in\oneto{N}$: \\
  \pcind - $\state_i=(\share{\bm{x}}_i,\share{\bm{\beta}}_i,\share{\bm{a}}_i\share{c}_i)$\\
  3. The prover computes the commitments: $\cmt_i=\commit{\state_i,\rho_i}$, where $\rho_i$ is a randomization sampled from $\{0, 1\}^\lambda$ for each $i\in\oneto{N}$.\\
  4. The prover computes and sends the Merkle tree root: $h_0 = \textsf{Merkle}(\cmt_1, \cdots, \cmt_N)$ \\[\baselineskip]
  \textbf{Step 2: First Challenge} \\
  5. The verifier sends $\big( (\gamma_j)_{j \in \oneto{n}}, \epsilon \big) \in \Fqme^n \times \Fqme$ to the prover.\\[\baselineskip]
  \textbf{Step 3: First Response} \\
 6. The prover chooses a public subset $S$ of parties such that $|S| = \ell + 1$: \\
 7. The prover executes the algorithm described in Fig.\ref{exec_pi_th} for the shares $i \in S$, and obtains $(\share{\bm{\alpha}}_i,\share{v}_i)_{i\in S}$.\\
  8. The prover computes the value  $h_1=\hash \left((\share{\boldsymbol{\alpha}}_i,\share{v}_i)_{i\in S}\right)$ \\[\baselineskip]
  \textbf{Step 4: Second Challenge} \\
 9. The verifier sends a subset $ I \subset \oneto{N} $ of parties such that $ |I| = \ell $ to the prover \\[\baselineskip]
 \textbf{Step 5: Second Response}\\
 10. The prover sends $(\state_i,\rho_i)_{i\in I}$, and the authentication path to these commitments, to allow the verifier to rebuild the Merkle tree. The verifier deduces a value $\Tilde{h}_0$.\\
 11. The prover chooses a public party $i^*\in S\backslash I$, and sends $\share{\boldsymbol{\alpha}}_{i^*}$ to allow the verifier to recover $\boldsymbol{\alpha}=\reconstruct_{I\bigcup \{i^*\}}(\share{\boldsymbol{\alpha}}_{I\bigcup \{i^*\}})$\\
 12. The verifier computes $\Tilde{h}_1$ by choosing $\share{v}_{i^*}$ such that $v=0$.\\
 13. The verifier outputs ACCEPT if $(\Tilde{h}_0,\Tilde{h}_1)=(h_0,h_1)$, REJECT otherwise.
}
\vspace{-\baselineskip}
\captionof{figure}{\footnotesize{MinRank Proof of Knowledge based on Threshold MPCitH}} \label{thmpcith}
\end{figure}

%% file: fig-exec_pi_th.tex
%!TEX root = ./main.tex

\begin{figure}[ht]
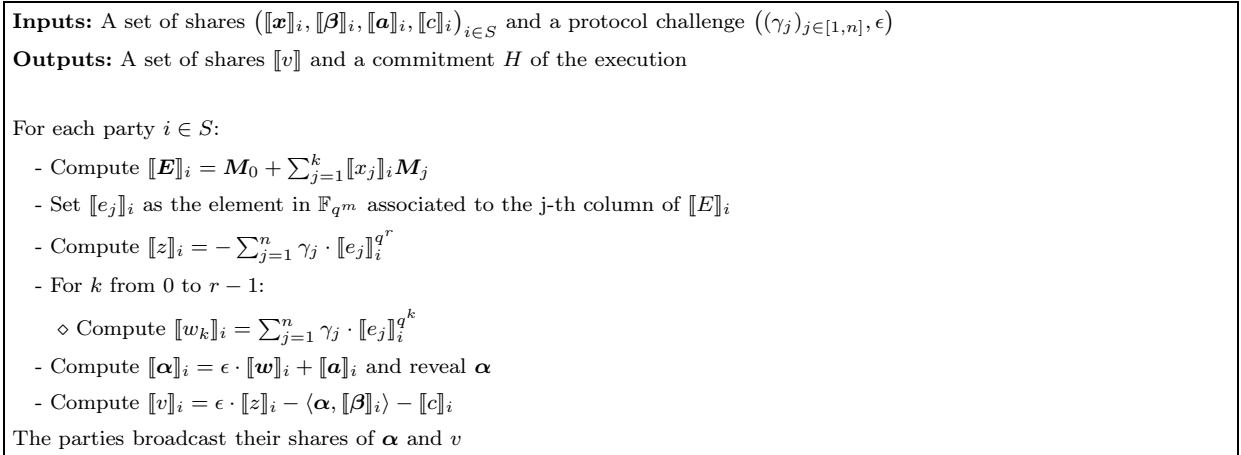

\pcb[codesize=\scriptsize, minlineheight=0.75\baselineskip, mode=text, width=0.98\textwidth] { 
  \textbf{Inputs:} A set of shares $\big(\share{\bm{x}}_i,\share{\bm{\beta}}_i,\share{\bm{a}}_i,\share{c}_i\big)_{i \in S}$ and a protocol challenge $\big( (\gamma_j)_{j \in \oneto{n}}, \epsilon \big)$\\
   \textbf{Outputs:} A set of shares $\share{v}$ and a commitment $H$ of the execution\\[\baselineskip]
  For each party $i\in S$: \\
\pcind - Compute $\share{\bm{E}}_i = \bm{M}_0 + \sum_{j=1}^{k}[\![x_j]\!]_i\bm{M}_j$ \\ 
  \pcind - Set  $[\![e_j]\!]_i$ as the element in $\Fqm$ associated to the j-th column of $[\![E]\!]_i$ \\
  \pcind - Compute $[\![z]\!]_i = -\sum_{j=1}^n \gamma_j \cdot [\![e_j]\!]_i^{q^r}$ \\
 \pcind - For $k$ from $0$ to $r-1$: \\
\pcind  \pcind $\diamond$ Compute $[\![w_k]\!]_i=\sum_{j=1}^n \gamma_j\cdot [\![e_j]\!]_i^{q^k}$ \\
  \pcind  - Compute $[\![\bm{\alpha}]\!]_i = \epsilon \cdot [\![\bm{w}]\!]_i+[\![\bm{a}]\!]_i$ and reveal $\bm{\alpha}$ \\
  \pcind  - Compute $[\![v]\!]_i=\epsilon \cdot [\![z]\!]_i-\langle\bm{\alpha},[\![\bm{\beta}]\!]_i\rangle-[\![c]\!]_i$ \\
 The parties broadcast their shares of $\bm{\alpha}$ and $v$
  }
\vspace{-\baselineskip}
\caption{\footnotesize{Execution of the MPC protocol $\Pi^\eta$ on a set of shares}} \label{exec_pi_th}
\end{figure}

%% file: fig-sign.tex
\begin{figure}[ht]
\pcb[codesize=\scriptsize, minlineheight=0.75\baselineskip, mode=text, width=0.98\textwidth] { 
\textbf{Inputs}\\
  \pcind - Public data $\bm{M}_0,\dots,\bm{M}_k \in \Fq^{m\times n}$\\
  \pcind - Secret key $\sk = (\bm{x}) \in \mathbb{F}_q^k$, $(\bm{E}) \in \Fq^{m\times n}$ such that $\operatorname{W}_R(\bm{E}) \leq r$ and $\bm{E} = \bm{M}_0 + \sum_{i=1}^{k}\bm{M}_ix_i$\\
  \pcind - Message $m \in \{0, 1\}^*$ \\[\baselineskip]
  \textbf{Step 1: Commitment} \\
  1. Sample a random salt value $\salt \sampler \{0, 1\}^{2\lambda}$ \\
  2. Compute $\bm{\beta} = (\beta_k)_{k \in [0, r - 1]} \in \Fqm^{r}$ the coefficients of the annihilator q-polynomial $L(X)$ associated to $\bm{{E}}$ such that $L(X) = X^{q^r} + \sum_{k = 0}^{r - 1} \beta_k (X^{q^k})$ and $\forall j \in \oneto{n}, L(\bm{{e}}_j) = 0$ where each $\bm{{e}_j}$ is the element in $\Fqm$ associated to the j-th column of $E$  \\
  3. For each iteration $e \in \oneto{\tau}$: \\
  \pcind - Sample $\bm{a}^{(e)} \sampler \Fqme^{r}$ and compute $c^{(e)}$ such that $c^{(e)} \in \Fqme$ and $c^{(e)} = - \dotp{\bm{\beta}, \bm{a}^{(e)}}$ \\
  \pcind - For each party $i \in \oneto{N}$: \\
  \pcind \pcind $\diamond$ Compute the $(\ell + 1, N)$-threshold sharings $\share{\bm{x}^{(e)}}, \share{\bm{\beta}^{(e)}}, \share{\bm{a}^{(e)}}, \share{c^{(e)}}$ of $\bm{x}, \bm{\beta}, \bm{a}^{(e)}$ and $c^{(e)}$ \\
  \pcind \pcind $\diamond$ Compute $\state^{(e)}_i = (\share{\bm{x}^{(e)}}_i, \share{\bm{\beta}^{(e)}}_i, \share{\bm{a}^{(e)}}_i, \share{c^{(e)}}_i)$ and $\cmt^{(e)}_i = \hash_0(\salt, e, i, \state^{(e)}_i)$ \\
  \pcind - Compute the Merkle tree root $h^{(e)}_0 = \textsf{Merkle}(\cmt^{(e)}_1, \cdots, \cmt^{(e)}_N)$ \\
  4. Compute $h_1 = \hash_1(m, \pk, \salt, (h^{(e)}_0)_{e \in \oneto{\tau}})$ \\[\baselineskip]
  \textbf{Step 2: First Challenge} \\
  5. Sample $\big( (\gamma^{(e)}_j)_{j \in \oneto{n}}, \epsilon^{(e)} \big)_{e \in \oneto{\tau}} \sampler PRG(h_1)$ where $ \big( (\gamma^{(e)}_j)_{j \in \oneto{n}}, \epsilon^{(e)} \big)_{e \in \oneto{\tau}} \in (\Fqme^n \times \Fqme)^{\tau}$ \\[\baselineskip]
  \textbf{Step 3: First Response} \\
  6. For each iteration $e \in \oneto{\tau}$: \\
  \pcind - For each party $i \in S$ with $S$ a public subset of parties such that $|S| = \ell + 1$: \\
  \pcind \pcind $\diamond$ Compute $\share{\bm{E}^{(e)}}_i = \bm{M_0} + \sum_{j=1}^k\bm{M_j}\share{x_j}_i$ and set $\share{e_j}_i$ as the element in $\Fqm$ associated to the j-th column of $\share{\bm{E}}$\\
  \pcind \pcind $\diamond$ Compute $\share{z^{(e)}}_i = - \sum\nolimits_{j = 1}^{n} \gamma_j \cdot \share{e_j^{(e)}}_i^{q^r}$ and $\forall k \in [0, r - 1], \share{w_k^{(e)}}_i = \sum\nolimits_{j = 1}^{n} \gamma_j \cdot \share{e_j^{(e)}}_i^{q^k}$ \\
  \pcind \pcind $\diamond$ Compute $\share{\boldsymbol{\alpha}^{(e)}}_i = \epsilon^{(e)} \cdot \share{\bm{w}^{(e)}}_i + \share{\bm{a}^{(e)}}_i$ and $\share{v^{(e)}}_i = \epsilon^{(e)} \cdot \share{z^{(e)}}_i - \dotp{\boldsymbol{\alpha}^{(e)}, \share{\bm{\beta}^{(e)}}_i} - \share{c^{(e)}}_i$ \\
  7. Compute $h_2 = \hash_2(m, \pk, \salt, h_1, (\share{\boldsymbol{\alpha}^{(e)}}_i, \share{v^{(e)}}_i)_{i \in S, e \in \oneto{\tau}})$ \\[\baselineskip]
  \textbf{Step 4: Second Challenge} \\
  8. Sample $(I^{(e)})_{e \in \oneto{\tau}} \sampler PRG(h_2) $ where $(I^{(e)})_{e \in \oneto{\tau}} \in  (\{ I \subset N ~|~ |I| = \ell \})^{\tau}$ \\[\baselineskip]
  \textbf{Step 5: Second Response} \\
  9. For each iteration $e \in \oneto{\tau}$: \\
  \pcind - Choose deterministically a party $i^{*(e)} \in S \, \backslash \, I^{(e)}$ and compute $\share{\boldsymbol{\alpha}^{(e)}}_{i^{*(e)}}$ \\
  \pcind - Compute the authentication path $\mathsf{auth}^{(e)}$ associated to root $h^{(e)}_0$ and $(\cmt^{(e)}_i)_{i \in I^{(e)}}$  \\
  \pcind - Compute $\rsp^{(e)} = \big( (\share{\bm{x}^{(e)}}_i, \share{\bm{\beta}^{(e)}}_i, \share{\bm{a}^{(e)}}_i, \share{c^{(e)}}_i)_{i \in I^{(e)}}, \mathsf{auth}^{(e)}, \share{\boldsymbol{\alpha}^{(e)}}_{i^{*(e)}} \big)$ \\
  10. Compute $\sigma = (\salt, h_1, h_2, (\rsp^{(e)})_{e \in \oneto{\tau}})$
}
\vspace{-\baselineskip}
\captionof{figure}{\footnotesize{MIRA Signature Scheme based on Threshold MPCitH - Signature Algorithm}}
\label{th_sig}
\end{figure}

%% file: fig-verify.tex
\begin{figure}[ht]
\pcb[codesize=\scriptsize, minlineheight=0.75\baselineskip, mode=text, width=0.98\textwidth] { 
  \textbf{Inputs} \\
  \pcind - Public data $\bm{M}_0,\dots,\bm{M}_k \in \Fq^{m\times n}$ such that there is $\bm{x} \in \Fq^k$ such that $\bm{E} = \bm{M}_0 + \sum_{i=1}^k(\bm{M}_ix_i) $ and $\operatorname{W}_R(\bm{E}) \leq r$\\
  \pcind - Message $m \in \{0, 1\}^*$ \\
  \pcind - Signature $\sigma = (\salt, h_1, h_2, (\bar{\rsp}^{(e)})_{e \in \oneto{\tau}})$ \\[\baselineskip]
  \textbf{Step 1: Parse signature} \\
  1. Sample $\big( (\gamma^{(e)}_j)_{j \in \oneto{n}}, \epsilon^{(e)} \big)_{e \in \oneto{\tau}} \sampler PRG(h_1)$ where $ \big( (\gamma^{(e)}_j)_{j \in \oneto{n}}, \epsilon^{(e)} \big)_{e \in \oneto{\tau}} \in (\Fqme^n \times \Fqme)^{\tau}$ \\
  2. Sample $(I^{(e)})_{e \in \oneto{\tau}} \sampler PRG(h_2)$ where $(I^{(e)})_{e \in \oneto{\tau}} \in  (\{ I \subset N ~|~ |I| = \ell \})^{\tau}$ \\
  3. For each iteration $e \in \oneto{\tau}$: \\
  \pcind - Choose deterministically $i^{*(e)}$ from $S \setminus I^{(e)}$ \\
  \pcind - Parse $\bar{\rsp}^{(e)} := \big( (\share{\bar{\bm{x}}^{(e)}}_i, \share{\bar{\bm{\beta}}^{(e)}}_i, \share{\bar{\bm{a}}^{(e)}}_i, \share{\bar{c}^{(e)}}_i)_{i \in \bar{I}^{(e)}}, \bar{\mathsf{auth}}^{(e)}, \share{\bar{\boldsymbol{\alpha}}^{(e)}}_{i^{*(e)}} \big)$ \\[\baselineskip]
  \textbf{Step 1: Recompute $\bm{h}_1$} \\
  4. For each iteration $e \in \oneto{\tau}$: \\
  \pcind - For each party $i \in I^{(e)}$: \\
  \pcind \pcind $\diamond$ Compute $\bar{\state}^{(e)}_i = (\share{\bar{\bm{x}}^{(e)}}_i, \share{\bar{\bm{\beta}}^{(e)}}_i, \share{\bar{\bm{a}}^{(e)}}_i, \share{\bar{c}^{(e)}}_i)$ and $\bar{\cmt}^{(e)}_i = \hash_0(\salt, e, i, \bar{\state}^{(e)}_i)$ \\
  \pcind - Compute the Merkle tree root $\bar{h}^{(e)}_0$ from $(\bar{\cmt}^{(e)}_i)_{i \in I^{(e)}}$ and $\bar{\mathsf{auth}}^{(e)}$ \\
  5. Compute $\bar{h}_1 = \hash_1(m, \pk, \salt, (\bar{h}^{(e)}_0)_{e \in \oneto{\tau}})$ \\[\baselineskip]
  \textbf{Step 2: Recompute $\bm{h}_2$} \\
  6. For each iteration $e \in \oneto{\tau}$: \\
  \pcind - For each party $i \in I^{(e)}$: \\
  \pcind \pcind $\diamond$ Compute $\share{\bm{E}^{(e)}}_i = \bm{M}_0 + \sum_{j=1}^k\bm{M}_j\share{x_j}_i$ and set $\share{e_j}_i$ as the element in $\Fqm$ associated to the $j$-th column of $\share{\bm{E}}$\\
  \pcind \pcind $\diamond$ Compute $\share{z^{(e)}}_i = - \sum\nolimits_{j = 1}^{n} \gamma_j \cdot \share{e_j^{(e)}}_i^{q^r}$ and $\forall k \in [0, r - 1], \share{w_k^{(e)}}_i = \sum\nolimits_{j = 1}^{n} \gamma_j \cdot \share{e_j^{(e)}}_i^{q^k}$ \\
  \pcind \pcind $\diamond$ Compute $\share{\boldsymbol{\alpha}^{(e)}}_i = \epsilon^{(e)} \cdot \share{\bm{w}^{(e)}}_i + \share{\bm{a}^{(e)}}_i$\\
  \pcind - Reconstruct $\bar{\boldsymbol{\alpha}}^{(e)}$ and $(\share{\bar{\boldsymbol{\alpha}}^{(e)}}_i)_{i \in S}$ from $(\share{\bar{\boldsymbol{\alpha}}^{(e)}}_i)_{i \in I^{(e)}}$ and $\share{\bar{\boldsymbol{\alpha}}^{(e)}}_{i^{*(e)}}$ \\
  7. For each iteration $e \in \oneto{\tau}$: \\
  \pcind - For each party $i \in I^{(e)}$: \\
  \pcind \pcind $\diamond$ Compute $\share{\bar{v}^{(e)}}_i = \epsilon^{(e)} \cdot \share{\bar{z}^{(e)}}_i - \dotp{\bar{\boldsymbol{\alpha}}^{(e)}, \share{\bar{\bm{\beta}}^{(e)}}_i} - \share{\bar{c}^{(e)}}_i$ \\
  \pcind - Reconstruct $(\share{\bar{v}^{(e)}}_i)_{i \in S}$ from $(\share{\bar{v}^{(e)}}_i)_{i \in I^{(e)}}$ and $\bar{v}^{(e)} = 0$ \\
  8. Compute $\bar{h}_2 = \hash_2(m, \pk, \salt, \bar{h}_1, (\share{\bar{\boldsymbol{\alpha}}^{(e)}}_i, \share{\bar{v}^{(e)}}_i)_{i \in S, e \in \oneto{\tau}})$ \\[\baselineskip]
  \textbf{Step 5: Verify signature} \\
  9. Return $(\bar{h}_1 = h_1) \wedge (\bar{h}_2 = h_2)$
}
\vspace{-\baselineskip}
\captionof{figure}{\footnotesize{MIRA Signature Scheme based on Threshold MPCitH - Verification Algorithm}}
\label{th_sig_verify}
\end{figure}

%% file: 4-1-parameters.tex
Our signature scheme uses the following parameters:
\begin{itemize}
\item the power of a prime number, $q \in \mathbb{N}$, to build $\Fq$; 
\item a positive integer, $m \in \mathbb{N}$, the number of rows of our matrices;
\item a positive integer, $n \in \mathbb{N}$, the number of columns of our matrices;
\item a positive integer, $k \in \mathbb{N}$, the length of the secret vector $\bm{x}$, and $k+1$ is the number of matrices in the public key;
\item a positive integer, $r \in \mathbb{N}$, the rank of the matrix $\bm{E}$;
\item a positive integer, $N \in \mathbb{N}$, the number of parties simulated in the MPC protocol;
\item a positive integer, $\eta \in \mathbb{N}$, to build $\Fqme$;
\item a positive integer, $\tau \in \mathbb{N}$, the number of rounds in the signature.
\end{itemize}
In the LSSS-based scheme, there is also a parameter $\ell \in \mathbb{N}$, which is the value of the privacy threshold of the used linear secret sharing scheme.

In order to choose the parameters, we need to consider:
\begin{itemize}
    \item The security of the MinRank instance, \textit{i.e.} the complexity of the attacks on the chosen MinRank parameters;
    \item The security of the signature scheme, \textit{i.e.} the cost of the best forgery attack;
    \item The size of the signature.
\end{itemize}

We propose the following parameters for the MinRank instance, %for a security level of $\lambda = 128$, i.e,
for the security level which corresponds to the NIST Security Level 1:
\begin{itemize}
    \item For the additive-based scheme: $(q,m,n,k,r) = (16,16,16,120,5)$;
    \item For the LSSS-based scheme: $(q,m,n,k,r) = (251,12,13,55,5)$.
\end{itemize}

In the first scheme, we take the value of $q=16$, since it leads to almost the shortest signature size and has the advantage to be easily serialized (two elements of $\Fq$ can be stored in a byte). However, in the second scheme, we take $q=251$ since we have the constraint $N \leq q$ with Shamir's secret sharing (since each share is associated to a distinct field element, we can not have more shares than the size of the field). Since the protocol is more efficient when $N$ is large, we need to take a larger value of $q$ in the threshold-based scheme compared to the additive-based scheme. The parameters $(m,n,k,r)$ are then chosen such that the MinRank instances are secure against the existing attacks, that we present shortly in section~\ref{atk_mr}, and such that $k+1 = (m-r)(n-r)$, as it corresponds exactly to the Gilbert-Varshamov bound.

Finally, we need to choose $\tau$ and $\eta$ such that our signature scheme resists to the forgery attack described by \cite{KZ20}. More precisely, in the additive-based scheme, we select them such that  
\begin{equation} \text{cost}_{\text{forge}}=\min_{0\leq\tau'\leq\tau}\left\lbrace\dfrac{1}{\sum_{i=\tau'}^\tau \binom{\tau}{i}p^i(1-p)^{\tau-i}}+N^{\tau-\tau'}\right\rbrace
\label{kz_eq_hyp}\end{equation} is higher than $2^\lambda$, with $p = \frac{2}{q^{m\eta}}-\frac{1}{q^{2m\eta}}$.
%The value depends on $\tau$ and $\eta$, we then need to take them such that the signature is safe, and as small as possible, given the parameters of the MinRank instance.
In the case of the threshold protocol, this formula becomes:
\begin{equation}
\text{cost}_{\text{forge}}=\min_{0\leq\tau'\leq\tau}\left\lbrace\dfrac{1}{\sum_{i=\tau'}^\tau \binom{\tau}{i}p^i(1-p)^{\tau-i}}+\binom{N}{\ell}^{\tau-\tau'}\right\rbrace \label{kz_eq_th}
\end{equation}
with $p=( \frac{2}{q^{m\eta}}-\frac{1}{q^{2m\eta}})\cdot \binom{N}{\ell+1}$. The factor $\binom{N}{\ell+1}$ in the value of $p$ comes from the fact that an attacker can commit an invalid sharing of the MPC inputs in the first step of the scheme (see \cite{FR22} for details), and the addition of $\binom{N}{\ell}^{\tau-\tau'}$ comes from the fact that there are $\binom{N}{\ell}$ possible second challenges.

In any case, as long as the forgery cost is high enough, we can just take the parameters which lead to the shortest signatures. To proceed, we need to compute the theoretical size of the signature, which we will do in the following section.

%% file: 4-2-hypercubesignsize.tex
For the additive-sharing protocol (Fig.\ref{add_sig}), we have to send $\tau$ times the following elements: 
\begin{itemize}
    \item $com_{i^*}\in\lbrace 0,1\rbrace^{2\lambda}$;
    \item $(state_{i'}^{(e)})_{i'\neq i^*}$;
    \item $[\![\bm{\alpha}]\!]\in\mathbb{F}_{q^{m\eta}}^{r}$;
    \item $[\![\bm{x}]\!]\in\mathbb{F}_q^k$;
    \item $[\![\bm{\beta}]\!]\in\Fqm^{r}$;
    \item $[\![c]\!]\in\mathbb{F}_{q^{m\eta}}$.
\end{itemize}
Note that for each $e\in \oneto{\tau}$, we do not send all the $N-1$ states $(state_{i'}^{(e)})_{i'\neq i^*}$. To reveal the state of all the parties except one, the prover only need to send the sibling path of the hidden leaf party in the seed tree of TreePRG. Thus, the number of revealed seeds in $\lbrace 0,1\rbrace^\lambda$ is equal to the depth of the tree, which is $D= \log_2 N$. We also add  $h_1,h_2, \text{ and } \salt \in \{0,1\}^{2\lambda}$, which add up to $6\lambda$. This gives us the following signature size (in bits):
\begin{align*}
\text{$|\sigma|$} &= \underbrace{6\lambda}_\text{$\salt,h1,h2$}+ \tau \cdot \left( \Big(\underbrace{k}_\text{$\bm{x}$}+\underbrace{r\times m}_\text{$\bm{\beta}$} + \underbrace{(r+1)\times m \times \eta}_\text{$\bm{\alpha}$~\text{and}~$c$}\Big)\cdot \log_2 q + \underbrace{2\lambda + D\lambda}_\text{additive MPCitH}\right).
\end{align*}

We then obtain the following sizes:
\begin{center}
{\scriptsize
    \begin{tabular}{|c|c|c|c|c|c|c|c|c|c|c|c|}
        \hline
        NIST Security Level & $q$ & $m$ & $n$ & $k$ & $r$& $N$ &$\tau$&$\eta$&\shortstack{Public Key size \\(Bytes)}& \shortstack{Secret Key size \\(Bytes)} & \shortstack{Signature size \\(Bytes)} \\ \hline
        1 & 16 & 16 & 16 & 120 & 5 & 256 & 18 & 1 & 84 & 16 & 5.640 \\ \hline
        3 & 16 & 19 & 19 & 168 & 6 & 256 & 26 & 1 & 121 & 24 & 11.779 \\ \hline
        5 & 16 & 23 & 22 & 271 & 6 & 256 & 34 & 1 & 150 & 32 & 20.762\\ \hline
    \end{tabular}
    \vspace{0.5\baselineskip}
}
\end{center}

%% file: 4-3-thresholdsignsize.tex
For the threshold sharing protocol (Fig.\ref{th_sig}), we have to send $\tau$ times the following elements: \begin{itemize}
    \item $\ell$ times $\share{\bm{x}}\in \Fqk$;
    \item $\ell$ times $\share{\bm{\beta}} \in \Fqm^{r}$;
    \item $\ell$ times $\share{\bm{a}} \in \Fqme^{r}$;
    \item $\ell$ times $\share{c} \in \Fqm$;
    \item $\share{\bm{\alpha}} \in \Fqme^{r}$;
    \item the authentification path $\mathsf{auth}$.
\end{itemize}
Let us remark that we cannot use a TreePRG for the threshold scheme, since the shares of a low-threshold sharing are correlated.

%If we used a TreePRG to generate seeds in the additive case, we do not have any interest to do this here: the TreePRG is efficient if you must reveal a large number of shares, but is absolutely not if you only need to reveal a few shares. Consequently, we don't send seeds, but immediately the $\ell$ shares you reveal. This implies that you don't have to send all the seed-uncles of the shares you want to reveal in the TreePRG, but you send $\ell$ of these elements.\\

For the authentification path, we do not need to send all the states. Indeed, to allow the verification of $\ell$ states, the prover sends their sibling paths which at most contain $\ell\cdot \log_2(\frac{N}{\ell})$ labels in total (each of $2\lambda$ bits). In the sizes we will give, we are going to use the average cost (with low deviation) of sending the sibling paths. We also send  $h_1,h_2, \text{ and } \salt \in \{0,1\}^{2\lambda}$, which add up to $6\lambda$. We find the following signature size (in bits): 
$$|\sigma| \leq \underbrace{6 \lambda}_{\salt,h1,h2} + \tau \cdot \Big( \big(\ell \cdot (\underbrace{k}_{\bm{x}} + \underbrace{r\times m}_\text{${\bm{\beta}}$} + \underbrace{(r+1)\times m\times\eta}_\text{$a$ and $c$}) + \underbrace{r\times m\times\eta}_\text{$\bm{\alpha}$}\big)\cdot \log_2(q) + \underbrace{2\lambda\cdot\ell \log_2(\frac{N}{\ell})}_\text{Threshold MPCitH} \Big).$$

We then obtain the following sizes:
\begin{center}
{\scriptsize
    \begin{tabular}{|c|c|c|c|c|c|c|c|c|c|c|c|c|}
        \hline
        NIST Security Level & $q$ & $m$ & $n$ & $k$ & $r$& $N$ &$\ell$&$\tau$&$\eta$&\shortstack{Public Key size \\(Bytes)}& \shortstack{Secret Key size \\(Bytes)} & \shortstack{Signature size \\(Bytes)} \\ \hline
        \multirow{1}*{1} & \multirow{1}*{251} & \multirow{1}*{12}& \multirow{1}*{13} & \multirow{1}*{55} & \multirow{1}*{5} & \multirow{1}*{251} &3 & 7 & 1 & \multirow{1}*{117} & \multirow{1}*{16} & 8.318\\ 
         \hline
         \multirow{1}*{3} & \multirow{1}*{251} & \multirow{1}*{16}& \multirow{1}*{15} & \multirow{1}*{109} & \multirow{1}*{5} & \multirow{1}*{251} &3 &  10& 1 & \multirow{1}*{155} & \multirow{1}*{24} & 17.797\\ 
         \hline
         \multirow{1}*{5} & \multirow{1}*{251} & \multirow{1}*{16}& \multirow{1}*{17} & \multirow{1}*{109} & \multirow{1}*{6} & \multirow{1}*{251} &3 & 14 & 1 & \multirow{1}*{195} & \multirow{1}*{32} & 30.381\\ 
         \hline
    \end{tabular}
    \vspace{0.5\baselineskip}
}    
\end{center}

%% file: 5-1-1-security_proof_hypercube.tex
%!TEX root = ./main.tex

\subsubsection{Security Proofs for the Additive Protocol} \hfill \\

We first need to prove that the proof of knowledge is sound and zero-knowledge:
\begin{theorem}
\label{Theoremhypercube}
The MinRank Proof of Knowledge protocol based on additive secret sharing described in Fig.\ref{hypercubempcith} has the following properties: \begin{itemize}
\item\textbf{Completeness:} A prover $\mathcal{P}$ who has the knowledge of a solution of a MinRank instance will always be accepted by the verifier. 
\item\textbf{Soundness:} Suppose that there is an efficient prover $\Tilde{\mathcal{P}}$ that convinces the verifier to accept with probability $$\Tilde{\epsilon} > \epsilon$$ with $$\epsilon = \frac{1}{N}+ p\cdot (1-\frac{1}{N})$$ where $\epsilon$ is the soundness of the protocol in Fig.\ref{hypercubempcith}, and $p$ is the false positive rate of the MPC protocol used, i.e, $\frac{2}{q^{m\eta}}-\frac{1}{q^{2m\eta}}$.\\
Then, there is an efficient probabilistic extraction algorithm, $\mathcal{E}$ that, given a rewindable black-box access to $\Tilde{\mathcal{P}}$, outputs either a solution of the MinRank instance, or a commitment collision by making a number of calls to $\Tilde{\mathcal{P}}$ which is bounded by $$\frac{4}{\Tilde{\epsilon}-\epsilon}\cdot \Big(1+\Tilde{\epsilon}\cdot\frac{2\operatorname{ln}(2)}{\Tilde{\epsilon}-\epsilon}\Big)$$
\item\textbf{Honest-Verifier Zero-Knowledge} If the pseudo-random generator algorithm PRG and the commitment $\mathsf{Com}$ are indistinguishable from the uniform random distribution, then the algorithm \ref{hypercubempcith} is Honest-Verifier Zero Knowledge.
\end{itemize}
\end{theorem}

\begin{proof} \hfill \\
\textbf{Completeness:}\\
By construction, if the prover has knowledge of a solution of the MinRank instance, he will always be able to execute the protocol correctly, i.e, he will always obtain $\bm{\alpha}$ such that $v=0$ when executing the MPC protocol $\Pi^\eta$, this is obvious. \\

\textbf{Soundness:}\\
We first need to establish that the probability for the malicious prover (who has no knowledge of the solution of the MinRank instance used and thus uses a bad witness) to cheat is at most $\epsilon = \frac{1}{N}+(1-\frac{1}{N})\cdot p$. \\
There are two situations where a malicious prover can be accepted by the verifier if he commits a bad witness: \begin{itemize}
    \item He obtains the value $v=0$ when executing the MPC protocol;
    \item The verifier believes that the value obtained $v$ is $0$.
\end{itemize}
We suppose here that the malicious prover commits a bad witness. Then, the first case occurs with probability $p = \frac{2}{q^{m\eta}}-\frac{1}{q^{2m\eta}}$ since it is the false positive rate of the protocol $\Pi^\eta$.

In the second case, the malicious prover needs to alter the communications in order to pass the verification. More precisely, he needs to alter the value of some share(s) of $\bm{\alpha}$ in every MPC protocol. Among the $N$ leaf shares, only one share, $i^*$, will not be revealed by the prover. This means that, if he cheats on more than one share, the verifier will notice the cheating, and thus rejects the proof. However, if he cheats on zero share, he will be rejected as well since the value $v$ will not be $0$ since the malicious prover doesn't have a good witness $\bm{x}$. This means the malicious prover can only cheat on one share exactly. However, cheating on one share means cheating on one main share on all the $D$ dimensions, as the main shares are the sum of leaves that have the same index $i_j$ along the current dimension. This means that the cheating is not detected if and only if the share the prover cheated on is $i^*$, since there exists a bijection between leaves and the set of their associated main party. This happens with probability $\frac{1}{N}$ (as the prover doesn't know the value of $i^*$ before cheating). \\
This is the only pattern to cheat and avoid detection (as shown just now, we cannot cheat on more than one leaf), since cheating on $1$ main party for each dimension is exactly equivalent to cheating on $1$ leaf party. \\
Thus, since there is no other cheating pattern possible, the probability for the malicious prover to be authenticated is at most $p+(1-p)\cdot \frac{1}{N} = \frac{1}{N}+(1-\frac{1}{N})\cdot p$.\\ [0.3cm]
We then need to show the soundness property in the theorem: \\
let $\mathcal{T}_1$ and $\mathcal{T}_2$ two transcripts with the same commitments, i.e, the same $h_0 = \hash({\cmt_1, \dots, \cmt_{N}})$ , but the second challenges $i^*_1$ (for $\mathcal{T}_1$) and $i^*_2$ (for $\mathcal{T}_2$) differ.

Then, we have two possibilities: \begin{itemize}
    \item $\share{\bm{x}}, \share{\bm{\beta}}, \share{\bm{a}}$ and $\share{c}$ differ in the two transcripts, and the malicious prover found a collision in the commitment hash;
    \item the openings of the commitments are equal, and thus the shares $\share{\bm{x}}, \share{\bm{\beta}}, \share{\bm{a}}$ and $\share{c}$ are equal in the transcripts.
\end{itemize}
We will only consider the second case, as we suppose that we use secure hash functions and secure commitment schemes.

Then, since $i^*_1$ and $i^*_2$ are different challenges and the commitments are the same, it is possible to recover the witness. We will then show we can build an extraction algorithm $\mathcal{E}$ that obtains a good witness.

$\bm{x}$ is called a good witness if it is a solution to the MinRank instance defined by public data, i.e. $\bm{E}= \bm{M}_0 + \sum_{i=1}^{k}\bm{M}_ix_i$ and $\operatorname{W}_R(\bm{E}) \le r$. Let $R_h$ the random variable associated to the randomness in initial commitment, and $r_h$ is the value it takes. For that, we will use the Splitting Lemma as is done in \cite{FJR22} and \cite{AGHHJY22}.

To get these two transcripts, the extraction algorithm $\mathcal{E}$ does the following: 
\begin{itemize}
    \item Run the protocol with randomness $r_h$ with the verifier until $\mathcal{T}_1$ is found, i.e, $\mathcal{T}_1$ is the first accepted transcript found by $\pt$. We note $i^*_1$ the leaf challenge obtained;
    \item Then, using the same randomness $r_h$ that was used, i.e, building the same commitments, $\mathcal{E}$ repeats the process $L$ times (where the value of $L$ is made explicit hereafter) until finding another accepted transcript, $\mathcal{T}_2$, for which the leaf challenge, $i^*_2$, is different than $i^*_1$; 
    \item If such a transcript $\mathcal{T}_2$ is found, then $\mathcal{E}$ recovers the witness, otherwise, if no such transcript is found after $L$ attempts, $\mathcal{E}$ returns to the first step and tries with another $r_h$.
\end{itemize}

In order to establish the soundness of the protocol, we need to estimate the number of times a malicious prover needs to repeat the authentication protocol in order to get the good witness $\bm{x}$. \\
Let $\delta \in ]0,1[$, and $\Tilde{\epsilon}$ such that $(1-\delta)\cdot\Tilde{\epsilon} > \epsilon$. We will define the randomness $r_h$ to be a good randomness if $\prb[\operatorname{Succ}_{\pt}\vert r_h ] > (1-\delta)\cdot\Tilde{\epsilon}$.

By the Splitting Lemma (Lemma \ref{splitting_lemma}), we have that $\prb [r_h\: good | \operatorname{Succ}_{\Tilde{\mathcal{P}}}] \ge \delta$. This means that after $\frac{1}{\delta}$ accepted transcripts, we have good odds to have a good randomness. Furthermore, we know that if the malicious prover uses a bad witness, his probability to cheat is bounded from above by $\epsilon$. Since the probability of success is greater than $\epsilon$, this means that a good witness has been used (when $r_h$ is good). 

To continue this proof, we will look at the probability to have, given an accepted transcript $\mathcal{T}_1$, a second accepted transcript, $\mathcal{T}_2$, with a challenge different than the one in $\mathcal{T}_1$. This means we are looking to bound from below the probability: 
$$ \prb[\operatorname{Succ}_{\Tilde{\mathcal{P}}} \cap (i^*_1 \ne i^*_2) | r_h \: good ]$$
Trivially, we know that this probability is equal to the probability of success knowing that $r_h$ is good, minus the probability of success with $i^*_1 = i^*_2$ knowing $r_h$ is good. This means: \small\begin{align*}   
\prb[\operatorname{Succ}_{\Tilde{\mathcal{P}}} \cap (i^*_1 \ne i^*_2) | r_h \text{ }good ] &=\prb[\operatorname{Succ}_{\Tilde{\mathcal{P}}} | r_h \text{ }good ] - \prb[\operatorname{Succ}_{\Tilde{\mathcal{P}}} \cap (i^*_1 = i^*_2) | r_h \text{ }good ]\\ & \ge \prb[\operatorname{Succ}_{\Tilde{\mathcal{P}}}|r_h \: good] - \frac{1}{N}\\ & \ge (1-\delta)\Tilde{\epsilon}-\frac{1}{N} \\ & \ge (1-\delta)\Tilde{\epsilon}-\epsilon\qquad \text{ (since $\epsilon \ge \frac{1}{N}$ trivially)}\end{align*}\normalsize

Now that we have this lower bound, we want to estimate the number of times one has to repeat the protocol to find $\mathcal{T}_2$. For that, we take the opposite probability, i.e, $1-\prb[\operatorname{Succ}_{\Tilde{\mathcal{P}}} \cap (i^*_1 \ne i^*_2) | r_h \text{ }good ]$, which is lower bounded by $1-((1-\delta)\Tilde{\epsilon}-\epsilon)$. We now want a probability of $\frac{1}{2}$ at least of success after $L$ tries of the authentication protocol. This means then that we want: \begin{align*}
\Big(1-\prb[\operatorname{Succ}_{\Tilde{\mathcal{P}}} \cap (i^*_1 \ne i^*_2) | r_h \text{ }good ]\Big)^L &< {\frac{1}{2}} \\
L \cdot \ln(1-((1-\delta)\Tilde{\epsilon}-\epsilon)) & < -\ln(2)\\
L &> -\frac{\ln(2)}{\ln(1-((1-\delta)\Tilde{\epsilon}-\epsilon))}\end{align*}
One obtains the following majoration for the number of calls to $\pt$:
$$L>\frac{\ln(2)}{\ln(\frac{1}{1-((1-\delta)\Tilde{\epsilon}-\epsilon)})} \approx \frac{\ln(2)}{(1-\delta)\Tilde{\epsilon}-\epsilon} $$
This means that, when repeating the protocol $L$ times, the probability to get the second transcript is higher than $\frac{1}{2}$.

Finally, we can look at the number of protocol repetitions that has to be done. 
To quickly remind the steps of the extraction: \begin{itemize}
    \item $\Tilde{\mathcal{P}}$ repeats the authentication protocol until he finds an accepted transcript $\mathcal{T}_1$, where the commitments are generated by $r_h$, and with second challenge $i^*_1$;
    \item When $\mathcal{T}_1$ is found, repeat the protocol with the same value $r_h$, $L$ times. After that, $\Tilde{\mathcal{P}}$ has more than $\frac{1}{2}$ chance of being successful. If he is not, he repeats from the first step of the procedure.
\end{itemize}

We will note $\mathbb{E}(\Tilde{\mathcal{P}})$ the number of calls the extractor $\mathcal{E}$ has to make to $\Tilde{\mathcal{P}}$. After $L$ calls (to find $\mathcal{T}_2$), if $r_h$ is good (which happens with probability $\delta$), we have $\frac{1}{2}$ chance of not finding $\mathcal{T}_2$. However, if $r_h$ is not good (with probability $1-\delta$), then we can consider that $\mathcal{T}_2$ is never found.

Thus, $\prb[\text{no }\mathcal{T}_2 | \operatorname{Succ}_{\Tilde{\mathcal{P}}}] = \frac{\delta}{2}+ (1-\delta) = 1-\frac{\delta}{2}$. If that happens, then, $\Tilde{\mathcal{P}}$ has to return to the first step, i.e, find $\mathcal{T}_1$ again. This means:$$ \mathbb{E}(\Tilde{\mathcal{P}}) \le 1+{\Big((1-\prb[\operatorname{Succ}_{\Tilde{\mathcal{P}}}])\mathbb{E}(\Tilde{\mathcal{P}}) \Big)} + \prb[\operatorname{Succ}_{\Tilde{\mathcal{P}}}]\Big( L+(1-\frac{\delta}{2})\mathbb{E}(\Tilde{\mathcal{P}})\Big)$$

Obviously, $\Tilde{\mathcal{P}}$ needs to run at least once. Then, we need to add to that the number of times expected before finding $\mathcal{T}_1$, and then, the number of times expected before finding $\mathcal{T}_2$. \\
Since $\prb[\operatorname{Succ}_{\Tilde{\mathcal{P}}}] = \Tilde{\epsilon}$ (by assumption), we can replace, and simplify the expression. We find then: \begin{align*}
    \mathbb{E}(\Tilde{\mathcal{P}}) & \le 1+{\Big((1-\Tilde{\epsilon}\cdot\mathbb{E}(\Tilde{\mathcal{P}}) \Big)} + \Tilde{\epsilon}\cdot\Big( L+(1-\frac{\delta}{2})\cdot\mathbb{E}(\Tilde{\mathcal{P}})\Big) \\ \mathbb{E}(\Tilde{\mathcal{P}}) & \le 1+\mathbb{E}(\Tilde{\mathcal{P}})-\Tilde{\epsilon}\cdot \mathbb{E}(\Tilde{\mathcal{P}}) + \Tilde{\epsilon} \cdot L + \Tilde{\epsilon} \cdot\mathbb{E}(\Tilde{\mathcal{P}}) - \Tilde{\epsilon} \cdot \frac{\delta}{2}\cdot \mathbb{E}(\Tilde{\mathcal{P}})\\
    \Tilde{\epsilon}\cdot\frac{\delta}{2}\cdot \mathbb{E}(\Tilde{\mathcal{P}}) &\le 1+ \Tilde{\epsilon}\cdot L \\
    \mathbb{E}(\Tilde{\mathcal{P}}) &\le \frac{2}{\Tilde{\epsilon}\cdot \delta}\Big(1+\Tilde{\epsilon}\cdot L\Big) = \frac{2}{\Tilde{\epsilon}\cdot \delta}\Big(1+\Tilde{\epsilon}\cdot \frac{\ln(2)}{(1-\delta)\Tilde{\epsilon}-\epsilon}\Big)
\end{align*}

Since this equality holds for any $\delta \in ]0,1[$, we can take $\delta$ such that $(1-\delta)\Tilde{\epsilon} = \frac{1}{2}(\Tilde{\epsilon}+\epsilon)$, and thus, we obtain the result: \begin{align*}
    \mathbb{E}(\Tilde{\mathcal{P}}) \le \frac{4}{\Tilde{\epsilon}-\epsilon}\cdot \Big(1+2\Tilde{\epsilon}\cdot \frac{\ln(2)}{\Tilde{\epsilon}-\epsilon}\Big)
\end{align*}
This means we found an upper bound on the number of calls the extractor $\mathcal{E}$ has to make to $\Tilde{\mathcal{P}}$ before retrieving a good witness, in the case where the probability to cheat was higher than $\epsilon$.\\[0.4cm]

\textbf{Honest-Verifier Zero-Knowledge:}\\
 Consider a simulator, described in Fig. \ref{hvzk_hyp}, which produces the transcript responses ($h_0$, $\ch_1$, $h_1$, $\ch_2$, $\rsp$).
We demonstrate that this simulator produces indistinguishable transcripts from the real distribution (the
one that we would obtain if it were generated by an honest prover who knows $\bm{x}$) by considering a succession of simulators: we begin by a simulator which produces true transcripts, and and change it gradually until arriving the following simulator. We explains why the distribution of transcripts is always the same at each step.
\input{fig-hvzk_hyp}
\begin{itemize}
    \item \textbf{Simulator 0 (real world):} it takes in input the witness $\bm{x}$ and the challenges $\ch_1 = \big( (\gamma_j)_{j \in \oneto{n}},\epsilon \big)$ and $\ch_2 = i^*$. It worrectly executes the algorithm 3, hence its output is the correct distribution.
    \item \textbf{Simulator 1:} Same as the \textbf{Simulator 0}, but uses true randomness instead of seed-derived randomness for leaf $i^*$.\\ If $i^* = N$, the leafs $\share{\bm{x}}_N,\share{\bm{\beta}}_N,\share{c}_N$ are computed as in the MPC protocol. \\
    The pseudo-random generator is supposed to be ($t$,$\epsilon_{PRG}$)-secured, its outputs are indistinguishable from the uniform distribution. Since the principal parts correspond to the sum of a certain number of leaves whose distributions are indistinguishable from that of the real worlds, their distribution is also indistinguishable.
    \item \textbf{Simulator 2:} Replace the leafs $\share{\bm{x}}_N,\share{\bm{\beta}}_N,\share{c}_N$ in \textbf{Simulator 1} by uniformly sampled values. Compute $\share{v}_{i^*} = - \sum_{i \ne i^*} \share{v}_i$. Note that this simulator becomes independent from the secret witness $\bm{x}$.\\
    If $i^* = N$, it only impacts the shares $\share{\bm{\alpha}}_{i^*}$ and $\share{v}_{i^*}$. Note that this change doesn't alter the uniform distribution of these values. It doesn't alter the distribution of any other leaf. \\
    If $i^* \ne N$, it only impacts $(\share{\bm{x}}_N,\share{\bm{\beta}}_N,\share{c}_N)$ in the simulated response and the values computed from them in the MPC protocol. It doesn't alter the distribution of other leaves. We observe that the shares in $(\share{\bm{x}}_N,\share{\bm{\beta}}_N,\share{c}_N)$ are calculated by adding a randomness value from each seed of party $i \ne i^*$, which correspond to adding a uniform random value from $\seed_{i^*}$. Since this distribution was uniform in \textbf{Simulator 1}, the output distributions are the same. Remember that this does not change the distributions of the main parts for the same reasons as before.
    \item \textbf{Simulator 3:} Rather than computing the value of $\share{\bm{\alpha}}_{i^*}$ as in the MPC protocol, sample it uniformly from $\Fqme^r$. As in the previous simulator, it doesn't change their output distribution.
\end{itemize}
As such, the output of the simulator is indistinguishable from the real distribution.

\end{proof}

%% file: fig-hvzk_hyp.tex
%!TEX root = ./main.tex

\begin{figure}[H]
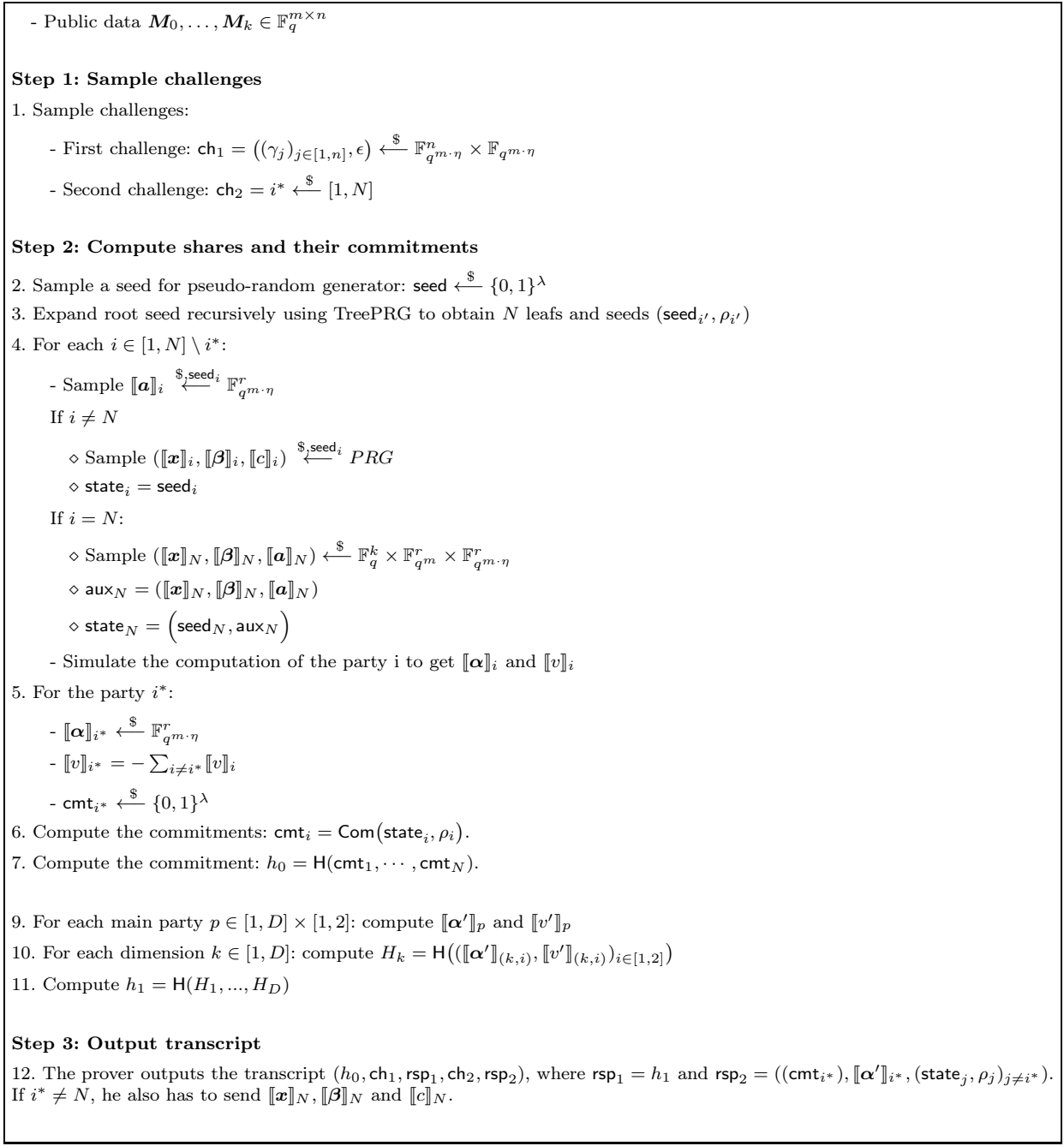

\pcb[codesize=\scriptsize, minlineheight=0.75\baselineskip, mode=text, width=0.98\textwidth] { 
    \pcind - Public data $\bm{M}_0,\dots,\bm{M}_k \in \Fq^{m\times n}$\\ 
   [\baselineskip]
  \textbf{Step 1: Sample challenges} \\
  1. Sample challenges:\\
  \pcind \pcind - First challenge: $\ch_1=\big( (\gamma_j)_{j \in \oneto{n}}, \epsilon \big)\sampler \Fqme^n \times \Fqme$ \\
  \pcind \pcind - Second challenge: $\ch_2= i^* \sampler\oneto{N}$ \\
  [\baselineskip]
  \textbf{Step 2: Compute shares and their commitments} \\
  2. Sample a seed for pseudo-random generator: $\seed\sampler\{0, 1\}^\lambda$\\
  3. Expand root seed recursively using TreePRG to obtain $N$ leafs and seeds $(\seed_{i'},\rho_{i'})$ \\
  4. For each $i\in\oneto{N}\setminus i^*$: \\
  \pcind \pcind - Sample $\share{\bm{a}}_i \samples{\seed_i} \Fqme^r$\\
  \pcind \pcind If $i \ne N$ \\
  \pcind \pcind \pcind $\diamond$ Sample $(\share{\bm{x}}_i,\share{\bm{\beta}}_i,\share{c}_i)\samples{\seed_i} PRG$\\
  \pcind \pcind \pcind $\diamond$ $\state_i=\seed_i$\\
  \pcind \pcind If $i = N$: \\
  \pcind \pcind \pcind  $\diamond$ Sample $\left(\share{\bm{x}}_{N},\share{\bm{\beta}}_{N},\share{\bm{a}}_{N}\right)\sampler\Fqk\times\Fqm^{r}\times\Fqme^{r}$\\
   \pcind \pcind \pcind $\diamond$ $\mathsf{aux}_N = \left(\share{\bm{x}}_{N},\share{\bm{\beta}}_{N},\share{\bm{a}}_{N}\right)$ \\
  \pcind \pcind \pcind $\diamond$ $\state_{N}=\Big (\seed_N,\mathsf{aux}_N \Big)$\\
  \pcind \pcind - Simulate the computation of the party i to get $\share{\bm{\alpha}}_i$ and $\share{v}_i$\\
  5. For the party $i^*$:\\
  \pcind \pcind - $\share{\bm{\alpha}}_{i^*} \sampler \Fqme^r$\\
  \pcind \pcind - $\share{v}_{i^*} = -\sum_{i \ne i^*} \share{v}_i$\\
  \pcind \pcind - $\cmt_{i^*} \sampler \{ 0,1\}^\lambda$\\
  6. Compute the commitments: $\cmt_i=\commit{\state_i,\rho_i}$.\\
  7. Compute the commitment: $h_0 = \hash(\cmt_1, \cdots, \cmt_{N})$.\\
  [\baselineskip]
  9. For each main party $p \in \oneto{D}\times \oneto{2}$: compute $\share{\bm{\alpha}'}_p$ and $\share{v'}_p$\\
  10. For each dimension $k\in\oneto{D}$: compute $H_k=\hash \big((\share{\bm{\alpha}'}_{(k,i)},\share{v'}_{(k,i)})_{i \in \oneto{2}}\big)$\\
  11. Compute $h_1=\hash (H_1,...,H_D)$\\
  [\baselineskip]
  \textbf{Step 3: Output transcript} \\
  12. The prover outputs the transcript $(h_0,\ch_1,\rsp_1,\ch_2,\rsp_2)$, where $\rsp_1=h_1$ and $\rsp_2=((\cmt_{i^*}),\share{\bm{\alpha}'}_{i^*},(\state_j,\rho_j)_{j \ne i^*})$. If $i^* \ne N$, he also has to send $\share{\bm{x}}_{N},\share{\bm{\beta}}_{N}$ and $\share{c}_{N}$.\\
}
\vspace{-\baselineskip}
\captionof{figure}{\footnotesize{HVZK simulator of the PoK with additive sharing and hypercube optimization}}
   \label{hvzk_hyp}
\end{figure}

%% file: 5-1-2-security_proof_th.tex
%!TEX root = ./main.tex

\subsubsection{Security Proofs for the Threshold Sharing Protocol} \hfill \\

\begin{theorem}
\label{Theoremth }
The MinRank Proof of Knowledge protocol based on threshold secret sharing described in Fig.\ref{thmpcith} has the following properties: \begin{itemize}
\item\textbf{Completeness:} A prover $\mathcal{P}$ who has the knowledge of a solution of a MinRank instance will always be accepted by the verifier. 
\item\textbf{Soundness:} Suppose that there is an efficient prover $\Tilde{\mathcal{P}}$ that convinces the verifier to accept with probability $$\Tilde{\epsilon} > \epsilon$$ with $$\epsilon = \frac{1}{\binom{N}{\ell}}+ p\cdot \frac{\ell\cdot(N-\ell)}{\ell+1}$$ where $\epsilon$ is the soundness of the protocol in Fig.\ref{thmpcith}, and $p$ is the soundness of the MPC protocol used, i.e, $\frac{2}{q^{m\eta}}-\frac{1}{q^{2m\eta}}$.\\
Then, there is an efficient probabilistic extraction algorithm $\mathcal{E}$ that, given a rewindable black-box access to $\Tilde{\mathcal{P}}$, outputs either a solution of the MinRank instance, or a commitment collision by making a number of calls to $\Tilde{\mathcal{P}}$ which is bounded by $$\frac{4}{\Tilde{\epsilon}-\epsilon}\cdot \Big(1+\Tilde{\epsilon}\cdot\frac{8\cdot(N-\ell)}{\Tilde{\epsilon}-\epsilon}\Big)$$
\item\textbf{Honest-Verifier Zero-Knowledge:} If the pseudo-random generator algorithm PRG and the commitment $\mathsf{Com}$ are indistinguishable from the uniform random distribution, then the algorithm \ref{thmpcith} is Honest-Verifier Zero Knowledge.
\end{itemize}
\end{theorem}

\begin{proof} \hfill \\
\textbf{Completeness:}\\
By construction, if the prover has knowledge of a solution of the MinRank instance, he will always be able to execute the protocol correctly, i.e, he will always obtain $\bm{\alpha}$ such that $v=0$ when executing the MPC protocol $\Pi^\eta$, this is obvious. \\

\textbf{Soundness:}\\

The proof is rather long and complex. For this proof, we refer to \cite{FR22}, who proved this theorem for any MPC protocol and MPCitH protocol, as long as they lie in their model. Our threshold protocol is an exact application of their model. Hence, the proof of the above theorem is the same as the proof in appendix D of \cite{FR22}.

\textbf{Honest-Verifier Zero-Knowledge:}\\

The proof is similar to what is done in the case of the additive sharing, and holds by the $t$-privacy of the MPC protocol, as well as the hiding property of the commitments. Moreover, a proof in the general case is done in \cite[Appendix C]{FR22}. Since we are in their model of MPCitH, the proof applies here as well.
\end{proof}

%% file: 5-2-1-security_proof_sig_hypercube.tex
%!TEX root = ./main.tex

\subsubsection{Security Analysis of MIRA-Additive}
We need to prove that the signature scheme is EUF-CMA secure:
\begin{theorem}
\label{secu_sig_hypercube}
    Let the PRG used be $(t,\epsilon_{PRG})$-secure, and $\epsilon_{MR}$ the advantage an adversary has over the MinRank problem. Consider $\hash_0,\hash_1,\hash_2,\hash_3,\hash_4$ behave as random oracles, with an output of $2\lambda$ bits. Then, if an adversary makes $q_i$ queries to $\hash_i$, $q_S$ queries to the signing oracle, the probability for him to produce a forgery (EUF-CMA) for the MIRA Additive Signature Scheme (Fig.\ref{add_sig}) is: 
    \begin{align*}
        \prb[\mathsf{Forge}] \le \frac{3\cdot (q + \tau \cdot N\cdot q_S)^2}{2\cdot 2^{2\lambda}}+\frac{q_S \cdot (q_S+5q)}{2^{2\lambda}}+ \epsilon_{PRG} + \prb[X+Y=\tau] + \epsilon_{MR}
    \end{align*}
    where $\tau$ is the number of rounds of the signature, $p=\frac{1}{q^{m\eta}}+\left(1-\frac{1}{q^{m\eta}}\right)\frac{1}{q^{m\eta}}$, $X = \operatorname{max}_{i\in [0,q_2]}\{X_i\}$ with $X_i \sim \mathcal{B}(\tau,p)$, and $Y = \operatorname{max}_{i\in [0,q_4]}\{Y_i\}$ with $Y_i \sim \mathcal{B}(\tau-X,\frac{1}{N})$.
\end{theorem}

\begin{proof}
    In this proof, we will adopt a game hopping strategy in order to find the upper bound. \\
    The first game will be the access to the standard signing oracle by the adversary $\mathcal{A}$. We will then game hop in order to eliminate the cases where collisions happen, and, through some other games, we will manage to find an upper bound. \\
    We note $\prb_i[\mathsf{Forge}]$ the probability of forgery when considering game $i$. 
    The aim of the proof is to find an upper bound on $\prb_1[\mathsf{Forge}]$.\\
    \begin{itemize}
    \item \textbf{Game 1} \\
    This is the interaction between $\mathcal{A}$ and the real signature scheme. \\
    $\mathsf{KeyGen}$ generates $(\bm{M}_0, \dots , \bm{M}_k,\bm{x})$ and $\mathcal{A}$ receives $\bm{M}_0 ,\dots,\bm{M}_k$. $\mathcal{A}$ can make queries to each $\hash_i$ independently, and can make signing queries. At the end of the attack, $\mathcal{A}$ outputs a message/signature pair, $(m,\sigma)$. The event $\mathsf{Forge}$ happens when $\sigma$ is a valid signature of $m$ and no signature of $m$ has been queried to the signing oracle. \\
    \item \textbf{Game 2} \\
    In this game, we add a condition to the success of the attacker. The condition we add is that if there is a collision between outputs of $\hash_0$, or $\hash_1$, or $\hash_3$, then, the forgery isn't valid. \\
    The first step is to look at the number of times every $\hash_i$ is called when calling the signing oracle. For $\hash_0$, we make $\tau \cdot N$ queries. The signing oracle contains also $\tau $ calls to $\hash_1$, one to $\hash_2$, $\tau\cdot D$ to $\hash_3$, and finally, a single one to $\hash_4$. \\
    The number of queries to $\hash_0$ or $\hash_1$ or $\hash_3$ is then bounded from above by $q+\tau \cdot N \cdot q_S$, where $q_i$ is the number of queries made by $\mathcal{A}$ to $\hash_i$, $q = \operatorname{max}\{q_0,q_1,q_2,q_3,q_4\}$ (we take $q_2$ and $q_4$ as well for $q$ since we are giving an upper bound), and $q_S$ is the number of queries to the signing oracle. \\
    We can then have the following result (it comes simply from the probability to have at least one collision with $q+\tau\cdot N\cdot q_S$ values):  
    $$\abs{\prb_1[\mathsf{Forge}]-\prb_2[\mathsf{Forge}]} \le \frac{3 \cdot (q+\tau\cdot N \cdot q_S)^2}{2\cdot 2^{2\lambda}}$$
   \item  \textbf{Game 3} \\
    The attacker now fails if the inputs to any of the $\hash_i$ has already appeared in a previous query. Wlog. the adversary does not make such a query itself and, if it happens in a signing query, this means that (at least) the salt randomly sampled by the signing oracle appears in a previous hash or signing query. We can bound this event with:  \begin{align*}\abs{\prb_2[\mathsf{Forge}]-\prb_3[\mathsf{Forge}]} &\le \frac{q_S \cdot (q_S + q_0 + q_1 + q_2 +q_3+q_4)}{2^{2\lambda}} &\le \frac{q_S \cdot (q_S + 5\cdot q)}{2^{2\lambda}}\end{align*}

  \item \textbf{Game 4}\\
  When beginning the signature of the message $m$, $h_1$ and $h_2$ are sampled uniformly and expanded into $(\gamma_1^{(e)}, \dots, \gamma_n^{(e)}, \epsilon^{(e)})_{e \in \oneto{\tau}}$, and $({i^{*}}^{(e)})_{e \in \oneto{\tau}}$. The game proceeds as before, but now we replace the queries to $\hash_{2}$ and $\hash_{4}$ by $h_1$ and $h_2$. If a query to $\hash_2$ or $\hash_4$ was already made, the attacker fails. However, this situation doesn't happen as \textbf{Game 3} would fail before. Hence, $$\prb_4[\mathsf{Forge}] = \prb_3[\mathsf{Forge}]$$
  
   \item  \textbf{Game 5} \\
    To answer the signing queries, we now use the \textbf{HVZK} simulator built in the previous proof, in order to generate the views of the open parties. By security of the PRG, the difference with the previous game is: 
    \begin{align*}\abs{\prb_5[\mathsf{Forge}]-\prb_4[\mathsf{Forge}]} &\le \epsilon_{PRG}\end{align*}
  \item   \textbf{Game 6} \\
    Finally, we say that an execution $e^*$ of a query $ h_2 = \hash_4(m,\salt,h_1,\share{H_1^{(e)}} \dots \share{H_D^{(e)}}_{e\in \oneto{\tau}})$ defines a good witness $\bm{x}$ if: \\
    - Each of the $H_k^{(e)}$ are the output of a query to $H_3$ \\
    - $h_1$ is the output of a query to $\hash_2$, i.e, $$h_1 = \hash_2(\salt,m,h_0^{(1)},\dots , h_0^{(\tau)})$$
    - Each $h_0^{(e)}$ is the output of a query to $\hash_1$, i.e, $$h_0^{(e)} = \hash_1(salt,e,\cmt_1^{(e)} ,\dots, \cmt_{N}^{(e)})$$
    -Each $\cmt_i^{(e)}$ is the output of a query to $\hash_0$, i.e, $$\cmt_i^{(e)} = \hash_0(\salt,e,i,\state_i^{(e)})$$
    - The vector $\bm{x} \in \Fqk$ defined by states $\{\state_i\}_{i \in \oneto{N}}$ is a correct witness, i.e, $\bm{E} = \bm{M}_0 + \sum_{i=1}^{k}\bm{M}_ix_i$ such that $\operatorname{W}_R(\bm{E}) \le r$. \\
    In case where such an execution happens, one can retrieve the correct witness from the states $\{\state_i\}_{i \in \oneto{N}}$ and, as a consequence, one can solve the MinRank instance. This means that $\prb_6[\mathsf{Solve}] \le \epsilon_{MR}$. \\[0.3cm]
    Finally, we only need to look at the upper bound of $\abs{\prb_6[\mathsf{Forge} \cap \mathsf{\overline{Solve}}]}$.
    This probability is upper bounded by the value $$\prb[X+Y=\tau]$$ where $X = \operatorname{max}_{i\in [0,q_2]}\{X_i\}$ with $X_i \sim \mathcal{B}(\tau,p)$, $Y = \operatorname{max}_{i\in [0,q_4]}\{Y_i\}$ with $Y_i \sim \mathcal{B}(\tau-X,\frac{1}{N})$. \\
    We explain this bound below: \\
    
    $\mathsf{Solve}$ doesn't happen here, meaning that, to have a forgery after a query to $\hash_4$, $\mathcal{A}$ has no choice but to cheat either on the first round or on the second one. \\
    
    \textit{Cheating at the first round.} For any query $Q_2$ to $\hash_2$, we call the output of this query $h_1$. For any query $Q_2$, if a false positive appears in a round $e$ with this value of $h_1$, then we add this round $e$ to the set we call $G_2(Q_2,h_1)$. This means that $\prb[e \in G_2(Q_2,h_1) \text{ }| \text{ } \mathsf{\overline{Solve}}] \le p=\frac{1}{q^{m\eta}}+\left(1-\frac{1}{q^{m\eta}}\right)\frac{1}{q^{m\eta}}$. Since the response $h_1$ is uniformly sampled, each round $e$ has the same probability to be in the set $G_2(Q_2,h_1)$. This means that $\#G_2(Q_2,h_1)$ follows the binomial distribution $X_{Q_2} = \mathcal{B}(\tau,p)$. We can then define $(Q_{2\text{best}},h_{1\text{best}})$ such that $\#G_2(Q_2,h_1)$ is maximized, i.e, $$\#G_2(Q_{2\text{best}},h_{1\text{best}}) \sim X = {\operatorname{max}\{X_{Q_2}\}}_{(Q_2 \in \mathcal{Q}_2)}$$

    \textit{Cheating at the second round.} Now, we need to look at the cheating in the second round, i.e, the queries to $\hash_4$. We will note this query $Q_4$, with the output of this query $h_2$.
    For the signature to be accepted, we know that, if in a round, the prover sends a wrong value of $h_1$, then he needs to cheat on exactly one leaf (it is already established that is isn't possible to cheat on less, or on more, than one leaf). He only needs to cheat when the value of $h_1^{(e)}$ is wrong, i.e, he needs to cheat for every round $e \notin G_2(Q_{2\text{best}},h_{1\text{best}})$. Since every time he cheats, the probability to be detected is $\frac{1}{N}$, it is easy to see the probability that the verification outputs ACCEPT is upper bounded by $\Big( \frac{1}{N}\Big)^{\tau-\#G_2(Q_{2\text{best}},h_{1\text{best}})} $ \\
    The probability that the prover is accepted on one of the $q_4$ queries is then upper bounded by $1-\Bigg(1- \Big( \frac{1}{N}\Big)^{\tau-\tau_1} \Bigg)$ where $\tau_1 =\#G_2(Q_{2\text{best}},h_{1\text{best}})$.
    By summing over all values of $\tau_1$ possible, we have then the upper bound: $$\prb_6[\mathsf{Forge} \cap \overline{\mathsf{Solve}}] \le \prb[X+Y = \tau]$$ where $X$ is as before, and $Y = {\operatorname{max}\{Y_{Q_2}\}}_{(Q_2 \in \mathcal{Q}_2)}$ where the $Y_{Q_2}$ are distributed following $\mathcal{B}(\tau-X,\frac{1}{N})$. 
    \end{itemize}
    All that is left to do is then to compute the sum of all the upper bounds we retrieved: this gives us the wanted result. \\
    
\end{proof}

%% file: 5-2-2-security_proof_sig_th.tex
%!TEX root = ./main.tex

\subsubsection{Security Analysis of MIRA-Threshold}

\begin{theorem}
\label{secu_sig_th}
    Let the PRG used be $(t,\epsilon_{PRG})$-secure, and $\epsilon_{MR}$ the advantage an adversary has over the MinRank problem. Consider $\hash_0,\hash_1,\hash_2$ and $\hash_M$ behave as random oracles, with an output of $2\lambda$ bits ($\hash_M$ is the function used for the Merkle Tree). Then, if an adversary makes $q_i$ queries to $\hash_i$, $q_S$ queries to the signing oracle, $q_M$ queries to $\hash_M$, the probability for him to produce a forgery (EUF-CMA) for the MIRA Threshold Signature Scheme (Fig.\ref{th_sig}) is: 
    \begin{align*}
        \prb[\mathsf{Forge}] \le \frac{(q + \tau \cdot (2 \cdot N-1) \cdot q_S)^2}{2^{2\lambda}}+\frac{q_S \cdot (q_S+3q)}{2^{2\lambda}}+ \epsilon_{PRG} + \prb[X+Y=\tau] + \epsilon_{MR}
    \end{align*}
    where $\tau$ is the number of rounds of the signature, $p=\frac{1}{q^{m\eta}}+\left(1-\frac{1}{q^{m\eta}}\right)\frac{1}{q^{m\eta}}$, $X = \operatorname{max}_{i\in [0,q_1]}\{X_i\}$ with $X_i \sim \mathcal{B}(\tau,\binom{N}{\ell+1}\cdot p)$, and $Y = \operatorname{max}_{i\in [0,q_2]}\{Y_i\}$ with $Y_i \sim \mathcal{B}(\tau-X,\frac{1}{\binom{N}{\ell}})$.
\end{theorem}

\begin{proof}
    In this proof, we will adopt a game hopping strategy in order to find the upper bound. \\
    The first game will be the access to the standard signing oracle by the adversary $\mathcal{A}$. We will then game hop in order to eliminate the cases where collisions happen, and, through some other games, we will manage to find an upper bound. \\
    We note $\prb_i[\mathsf{Forge}]$ the probability of forgery when considering game $i$. 
    The aim of the proof is to find an upper bound on $\prb_1[\mathsf{Forge}]$.\\
    \begin{itemize}
    \item \textbf{Game 1} \\
    This is the interaction between $\mathcal{A}$ and the real signature scheme. \\
    $\mathsf{KeyGen}$ generates $(\bm{M}_0, \dots , \bm{M}_k,\bm{x})$ and $\mathcal{A}$ receives $\bm{M}_0 ,\dots,\bm{M}_k$. $\mathcal{A}$ can make queries to each $\hash_i$ independently, and can make signing queries. At the end of the attack, $\mathcal{A}$ outputs a message/signature pair, $(m,\sigma)$. The event $\mathsf{Forge}$ happens when the message output by $\mathcal{A}$ was not previously used in a query to the signing oracle. \\
   \item  \textbf{Game 2} \\
    We add a condition to the success of the attacker now. If there is a collision in the outputs of $\hash_0$ or on $\hash_M$, then the forgery isn't valid. Here, $\hash_0$ is called $q_0$ times by $\mathcal{A}$, $\hash_M$ $q_M$ times. When $\mathcal{A}$ calls the signing oracle, there are in total: $\tau \cdot N$ calls to $\hash_0$, $\tau \cdot (2\cdot N-1)$ calls to $\hash_M$, and one to $\hash_1$ and $\hash_2$. In this game, only $\hash_0$ and $\hash_M$ are of interest, but for a simpler notation, we will take $q = \operatorname{max}\{ q_0,q_1,q_2,q_M\}$ (as it is an upper bound we are looking for, this is fine). We can then give an upper bound to the queries made by $\mathcal{A}$ to the hash functions, which is then: $q + \tau \cdot (2\cdot N-1)\cdot q_S$ where $q_S$ is the number of queries to the signing oracle. \\
    When making this many queries, we can now bound from above the probability of having a collision, with  $$\abs{\prb_1[\mathsf{Forge}]-\prb_2[\mathsf{Forge}]} \le \frac{(q + \tau \cdot (2\cdot N-1)\cdot q_S)^2}{2^{2\lambda}}$$
    \item \textbf{Game 3} \\
    The attacker now fails if the inputs to any of the $\hash_i$ has already appeared in a previous query. Wlog. the adversary does not make such a query itself and, if it happens in a signing query, this means that (at least) the salt randomly sampled by the signing oracle appears in a previous hash or signing query. Since we don't use $\salt$ in the Merkle Tree, this will only concern $\hash_0,\hash_1,\hash_2$. We sample $\salt$ $q_S$ time (once by signing oracle query), and $3\cdot q$ times as well (each time we call $\hash_0,\hash_1$ or $\hash_2$). If there is an input which already appears for $\hash_M$, this must be because a collision has been found either on $\hash_M$ or on $\hash_0$. However, we already excluded this in \textbf{Game 2}. This means we can give the following bound: \\
    \begin{align*}\abs{\prb_2[\mathsf{Forge}]-\prb_3[\mathsf{Forge}]} &\le \frac{q_S \cdot (q_S + q_0 + q_1 + q_2)}{2^{2\lambda}} &\le \frac{q_S \cdot (q_S + 3\cdot q)}{2^{2\lambda}}\end{align*}

  \item \textbf{Game 4}\\
  When beginning the signature of the message $m$, $h_1$ and $h_2$ are sampled uniformly and expanded into $\gamma_1^{(e)}, \dots, \gamma_n^{(e)}, \epsilon^{(e)}$, and ${i^{*}}^{(e)}$. The game proceeds as before, but now we replace the queries to $\hash_{2}$ and $\hash_{4}$ by $h_1$ and $h_2$. If a query to $\hash_2$ or $\hash_4$ was already made, the attacker fails. However, this situation doesn't happen as \textbf{Game 3} would fail before. Hence, $$\prb_4[\mathsf{Forge}] = \prb_3[\mathsf{Forge}]$$
 
  \item   \textbf{Game 5} \\
    To answer the signing queries, we now use the \textbf{HVZK} simulator built in the previous proof, in order to generate the views of the open parties. By security of the PRG, the difference with the previous game is: 
    \begin{align*}\abs{\prb_5[\mathsf{Forge}]-\prb_4[\mathsf{Forge}]} &\le \epsilon_{PRG}\end{align*}
 \item    \textbf{Game 6} \\
    Finally, we say that an execution $e^*$ of a query $ h_2 = \hash_2(m,\mathsf{pk},\salt,h_1, (\share{\boldsymbol{\alpha}^{(e)}}_i, \share{v^{(e)}}_i)_{i \in S, e \in \oneto{\tau}})$ defines a good witness $\bm{x}$ if: \\
    - $h_1$ is the output of a query to $H_1$, i.e, $$h_1 = \hash_1(\salt,m,h_0^{(1)},\dots , h_0^{(\tau)})$$
    - Each $h_0^{(e)}$ is the output of a query to the MerkleTree oracle, i.e, $$h_0^{(e)} = \mathsf{Merkle}(\cmt_1^{(e)} ,\dots, \cmt_{N}^{(e)})$$
    -Each $\cmt_i^{(e)}$ is the output of a query to $\hash_0$, i.e, $$\cmt_i^{(e)} = \hash_0(\salt,e,i,\state_i^{(e)})$$
    - The vector $\bm{x} \in \Fqk$ defined by states $\{\state_i\}_{i \in \oneto{N}}$ is a correct witness, i.e, $\bm{E} = \bm{M}_0 + \sum_{i=1}^{k}\bm{M}_ix_i$ such that $\operatorname{W}_R(\bm{E}) \le r$. \\
    In case where such an execution happens, one can retrieve the correct witness from the states $\{\state_i\}_{i \in \oneto{N}}$ and, as a consequence, one can solve the MinRank instance. This means that $\prb_6[\mathsf{Solve}] \le \epsilon_{MR}$. \\[0.3cm]
    Finally, we only need to look at the upper bound of $\abs{\prb_6[\mathsf{Forge} \cap \mathsf{\overline{Solve}}]}$.
    This probability is upper bounded by the value $$\prb[X+Y=\tau]$$ with $X = \operatorname{max}_{i\in [0,q_1]}\{X_i\}$ with $X_i \sim \mathcal{B}(\tau,\binom{N}{\ell+1}\cdot p)$, $Y = \operatorname{max}_{i\in [0,q_2]}\{Y_i\}$ with $Y_i \sim \mathcal{B}(\tau-X,\frac{1}{\binom{N}{\ell}})$ and where $p=\frac{1}{q^{m\eta}}+\left(1-\frac{1}{q^{m\eta}}\right)\frac{1}{q^{m\eta}}$. \\
    This result comes directly from \cite[Lemma 6 and Theorem 4, Appendix F]{FR22}.
    \end{itemize}
    All that is left to do is then to compute the sum of all the upper bounds we retrieved: this gives us the wanted result. \\

\end{proof}

%% file: 6-1-attacksFS.tex
There are several attacks against signatures from zero-knowledge proofs obtained thanks to the Fiat-Shamir heuristic. \cite{AAB} proposes an attack more efficient than  brute force for protocols with more than one challenge, i.e. for protocols of a minimum of 5 rounds.

Kales and Zaverucha proposed in \cite{KZ20} a forgery attack which consists in guessing separately the two challenges of the protocol. It results an additive cost rather than the expected multiplicative cost. The cost to forge a valid transcript for a 5-round proof of knowledge corresponds to the cost of the optimal trade-off between the work needed to pass the first step and the work needed to pass the second step. To run the attack, one can find the optimal number of repetitions for the brute-force work of the first step with the formula: $$\tau'=\arg\min_{0\leq \tau' \leq \tau}\left\{\frac{1}{\sum_{i=\tau'}^\tau \binom{\tau}{i}P_1^i(1-P_1)^{\tau-i}}+\Big(\frac{1}{P_2}\Big)^{\tau-\tau'}\right\}$$ where $P_1$ and $P_2$ are the probabilities to pass respectively the first and the second challenges for one repetition.

\subsubsection{Cost of forgery of MIRA-Additive} \hfill

In the additive case, one obtains:
$$\text{cost}_{\text{forge}}=\min_{0\leq\tau'\leq\tau}\left\lbrace\dfrac{1}{\sum_{i=\tau'}^\tau \binom{\tau}{i}p^i(1-p)^{\tau-i}}+(N)^{\tau-\tau'}\right\rbrace$$ where $p = \frac{2}{q^{m\eta}}-\frac{1}{2q^m\eta}$. 

\subsubsection{Cost of forgery of MIRA-Threshold} \hfill

In the threshold case, one obtains:
$$\text{cost}_{\text{forge}}=\min_{0\leq\tau'\leq\tau}\left\lbrace\dfrac{1}{\sum_{i=\tau'}^\tau \binom{\tau}{i}p'^i(1-p')^{\tau-i}}+\binom{N}{\ell}^{\tau-\tau'}\right\rbrace \label{kz_eq_th_1}$$ where $p' = \Big(\frac{2}{q^{m\eta}}-\frac{1}{2q^m\eta} \Big) \binom{N}{\ell+1}$.

%% file: 6-2-attacksMinRank.tex
In this section, we briefly describe the most effective attacks on MinRank. A reader can refer to \cite{BB22}, \cite{GNS23},\cite{BBBGT22}, \cite{BBCGPTTV20}, and \cite{GB00} for more details on the attacks.

To begin with, the following hybrid approach can benefit to all the other attacks. \\

\textbf{Hybrid approach}

\cite{BBBGT22} introduced a generic approach to improve all the attacks on MinRank. The idea of the attack is to solve smaller instances of MinRank instead.
The complexity is given by
\begin{align}
  \min_{a}\left(q^{ar}{\mathbb C}_{\mathcal A}(q,m,n-a,K-am,r)\right)
\end{align}
where $\mathbb C_{\mathcal A}(q,m,n,K,r)$ is the cost of an algorithm $\mathcal A$ to solve a MinRank instance. \\

\subsubsection{The Kernel Attack} \hfill \\

The kernel attack was described by Goubin and Courtois in \cite{GB00}. The idea of the attack is to take random vectors, and hoping that they are in the kernel of $\bm{E}$. Since $\bm{E}$ is of size $m \times n$, and is of rank at most $r$, $\operatorname{Ker}(\bm{E})$ will be a matrix of dimensions $n \times (n-r)$ at least. Thus, if $v \sampler \mathbb{F}_q^n$, $\prb[v \in \operatorname{Ker}(\bm{E})] = \frac{q^{n-r}}{q^n} = \frac{1}{q^r}$. Then, if we get $l$ independant vectors in $\operatorname{Ker}(\bm{E})$, and set the matrix $\bm{X}$ whose columns are the $l$ vectors, we can compute $(\bm{M}_0 + \sum_{i=1}^{k}x_i\bm{M}_i)X$, which gives us a linear system in $x_1 \dots x_k$, and $m\cdot l$ equations (since we have $(\bm{M}_0 + \sum_{i=1}^{k}x_i\bm{M}_i)X = 0$). Thus with $l = \lceil \frac{k}{m} \rceil$, we have a unique solution to the system, which we can find with linear algebra.\\
As to the complexity of the attack, it is quite obvious that it is in $O(q^{r\lceil \frac{k}{m} \rceil})$ to find the vectors in the kernel, and in $O(k^\omega)$ to solve the linear system. Hence, the total complexity is $$O(q^{r\lceil \frac{k}{m} \rceil}k^\omega)$$\\

\subsubsection{Algebraic Attacks} \hfill \\

\input{algebraic-MinRank}

%% file: algebraic-MinRank.tex
\newcommand{\eqdef}{\stackrel{\text{def}}{=}}
\newcommand{\any}{*} % any entry of matrix, typically for submatrices
\newcommand{\minor}[2]{\left\vert #1 \right\vert_{#2}}

\newcommand{\trsp}[1]{{#1}^{\top}}
\newcommand{\mb}[1]{\textcolor{blue}{\bf Magali: #1}}
% comments

\textbf{Minors Modeling}

The modeling was introduced and studied in \cite{FSS10} and \cite{FSS13}. This modeling uses the minors of the matrix $\bm{E}$, where the $x_i$ are still unknowns. It was also improved in \cite{GNS23}. We refer to these papers for the complexity of the attack.

\textbf{Support Minors Modeling}

The Support Minors modeling was introduced in~\cite{BBCGPTTV20}. This idea also uses minors of a matrix, giving us an other system of equations. With this approach, the complexity is of \begin{align*}
    \mathcal{O}\left( N_b {M_b }^{\omega-1}\right),
  \end{align*}
  where \begin{align}
  N_b &=  \sum_{i = 1}^{b}(-1)^{i+1} \binom{n}{ r+i}\binom{k+b-1-i}{b-i}
    \binom{m+i-1}{i}.\\
  M_b &= \binom{k+b-1}{b}\binom{n}{r}.
\end{align} and $b$ is the degree to which we augment the Macaulay matrix of the system (\cite{BBBGT22}). \\

%%% Local Variables:
%%% mode: latex
%%% TeX-master: "main"
%%% End: